\documentclass[]{article}

\usepackage[top = 0.75in, right = 0.75in, left = 0.75in, bottom = 0.75in]{geometry}
\usepackage{siunitx}
\usepackage{amsmath}
\usepackage{mathrsfs}
\usepackage{amssymb}
\usepackage{mathtools}
\usepackage{graphicx}
\usepackage{textcomp}
\usepackage{epstopdf}
\usepackage{epsfig}
\usepackage{rotating}
\usepackage{pdflscape}
\usepackage{afterpage}
\usepackage{tabulary}
\usepackage{pdfpages}
\usepackage{endnotes}
\usepackage{lscape}
\usepackage{longtable}
\usepackage{multirow}
\usepackage{xcolor}
\usepackage{url}
\usepackage{float}
\usepackage{stmaryrd}
\usepackage{amsthm}
\usepackage{tablefootnote}
\usepackage{enumitem}
\usepackage{nomencl}
\usepackage{authblk}

\newtheorem{assumption}{Assumption}
\newtheorem{definition}{Definition}
\newtheorem{algorithm}{Algorithm}
\newtheorem{theorem}{Theorem}

\newtheorem*{remark}{Remark}

\title{Distributed Economic Model Predictive Control -- Addressing Non-convexity Using Social Hierarchies}
\author[1]{Ali C. Kheirabadi}
\author[1]{Ryozo Nagamune}
\affil{The University of British Columbia, Vancouver Campus, 2054-6250 Applied Science Lane, Vancouver, BC Canada V6T 1Z4}

\begin{document}

\maketitle

\begin{abstract}
	
This paper introduces a novel concept for addressing non-convexity in the cost functions of distributed economic model predictive control (DEMPC) systems. Specifically, the proposed algorithm enables agents to self-organize into a hierarchy which determines the order in which control decisions are made. This concept is based on the formation of social hierarchies in nature. An additional feature of the algorithm is that it does not require stationary set-points that are known \textit{a priori}. Rather, agents negotiate these targets in a truly distributed and scalable manner. Upon providing a detailed description of the algorithm, guarantees of convergence, recursive feasibility, and bounded closed-loop stability are also provided. Finally, the proposed algorithm is compared against a basic parallel distributed economic model predictive controller using an academic numerical example.

\end{abstract}

\section{Introduction}

\subsection{Background}

Model predictive control (MPC) entails recursively solving an optimization problem over a finite prediction horizon to identify optimal future control input trajectories. The popularity of MPC in academic and industrial environments is primarily attributed to its capacity for handling constraints while computing control actions that minimize nonlinear performance criteria. The reader may refer to articles by Mayne \textit{et al.}~\cite{Mayne2000} and Mayne~\cite{Mayne2014} for reviews on MPC.

In large-scale processes or multi-agent systems, implementation of MPC in a centralized manner may be impractical due to the computational complexity of global optimization and the network infrastructure required for plant-wide communication. Distributed model predictive control (DMPC) surpasses these limitations by dispersing the burden of decision-making across a multitude of independent subsystems or agents. A trade-off that arises however, is that effective algorithms governing agent coordination are required to guarantee desirable closed-loop performance. The reader may refer to review articles by Al-Gherwi \textit{et al.}~\cite{Al-Gherwi2011a}, Christofides \textit{et al.}~\cite{Christofides2013}, and Negenborn and Maestre~\cite{Negenborn2014} for further details on the subject of DMPC.

MPC has traditionally been utilized as a lower-level regulator and stabilizer that tracks set-points determined by upper-level stationary optimizers. Economic model predictive control (EMPC) combines these upper- and lower-level roles by employing cost functions that capture plant economics (\textit{e.g.} power production or operating cost over a finite time horizon). The effect is improved economic performance; however, additional measures for ensuring stability are required since the primary control objective no longer involves regulation. The reader may refer to articles by Ellis \textit{et al.}~\cite{Ellis2014} and M\"uller and Allg\"ower~\cite{Muller2017} for reviews on EMPC.

\subsection{Distributed economic model predictive control}

This paper addresses distributed economic model predictive control (DEMPC) of systems with non-convex objective functions and unknown stationary set-points. Applications with such characteristics include autonomous vehicle trajectory planning~\cite{Eilbrecht2019} and floating offshore wind farm control~\cite{Kheirabadi2020}. DEMPC algorithms intended for such systems have been scarce in the literature as a result of challenges pertaining to stability and convergence. This subsection reviews relevant DEMPC and nonlinear DMPC algorithms as justification for the contributions of the current work.

\subsubsection{Stabilizing DEMPC algorithms}

Achieving stability in DEMPC requires first computing optimal stationary set-points for all agents, and then constraining state trajectories to approach these optima within the prediction horizon. If there exist feedback control laws that are then capable of maintaining subsystems within specified bounds of their respective steady-states, stability may be guaranteed. To achieve such an outcome, theoretical studies focused on DEMPC have either treated these stationary set-points as predefined references~\cite{Wang2017a}, or computed their values using centralized optimization~\cite{Lee2011,Lee2012,Driessen2012a,Chen2012,Albalawi2017a,Wolf2012}. The latter group of algorithms are therefore not truly distributed.

To overcome this gap, K\"ohler \textit{et al.}~\cite{Kohler2018b} were the first to develop a DEMPC scheme without the requirement for centralized processing. They presumed that optimal stationary set-points were unattainable via centralized optimization, and instead had to be negotiated online between agents in a distributed manner. Consequently, in tandem with solving their local EMPC problems and obtaining optimal input trajectories, agents also performed one iterate of a distributed coordination algorithm at each sampling time to update their respective optimal steady-states. Nonetheless, this work focused on linear systems with convex cost functions and used a sequential coordination algorithm~\cite{Kuwata2007,Richards2007}; thus suffering from lack of scalability.

\subsubsection{Convergent DEMPC algorithms}

If DEMPC cost functions are non-convex, agents making decisions in parallel cannot guarantee convergence of their optimal input trajectories~\cite{Liu2012a}. Several alternative classes of coordination algorithms within the nonlinear DMPC literature address this convergence issue. Sequential methods first proposed by Kuwata \textit{et al.}~\cite{Kuwata2007} and Richards and How~\cite{Richards2007} represent the simplest solution. Agents solve their local optimization problems and exchange information with their neighbors in some predetermined order. The resulting advantage is that each subsequent agent computes its input trajectory based on updated and fixed information from its predecessors; guaranteeing convergence, stability, and feasibility is thus facilitated. The major drawback is lack of scalability to large interconnected systems, since agents at the tail-end of sequence must await decisions from all other subsystems. A secondary concern involves predetermining the sequence order, particularly in systems with time-varying interaction topologies.

Coordination algorithms based on negotiation between agents were developed by M\"uller \textit{et al.}~\cite{Muller2012a}, Maestre \textit{et al.}~\cite{Maestre2011}, and Stewart \textit{et al.}~\cite{Stewart2011}. An agent receives optimal decisions from its neighbors in the form of a proposal. Then, upon computing the corresponding effects of these decisions on its local objective function, the agent may reject or approve proposals. These algorithms are capable of resolving conflict; however they face two limitations. The first is that, in order to identify the impact of a specific agent's control trajectory on neighboring cost functions, this agent must not operate in parallel with others; thus limiting scalability. The second is that agents whose control actions are discarded at particular time-steps remain idle. Finally, these algorithms possess no learning mechanism to ensure that, after a sufficient number of negotiations, proposals are guaranteed or more likely to be approved.

Finally, group-based DMPC methods employ the connectivity information of a plant to identify the order in which agents should solve their local MPC problems to resolve conflict. Pannek~\cite{Pannek2013} proposed a covering algorithm that permitted non-interacting agent pairs to operate in parallel, while those that were coupled made decisions sequentially according to some predetermined priority rule. This algorithm eliminated the scalability issue of pure sequential DMPC; however it required a predetermined set of priority rules. Liu \textit{et al.}~\cite{Liu2019a} developed a clustering algorithm that assigned agents to dominant or connecting groups. Agents in dominant clusters solved their local optimization problems first, thus eliminating conflict with agents in connecting groups. The downside in this method was that a sequential algorithm was required to determine clustering. Asadi and Richards~\cite{Asadi2018} employed a slot allocation algorithm wherein each agent communicated with other subsystems to randomly select an available space in the global sequential order. This method addressed the secondary drawback of sequential DMPC, which concerned determining an effective sequence order in systems with time-varying interaction topologies. Nonetheless, the fully serial nature of the algorithm still suffered from lack of scalability.

\subsection{Contributions}

Based on the preceding literature review, we state, to the best of our knowledge, that a DEMPC algorithm that handles non-convex cost functions and unknown stationary set-points in a scalable and truly distributed manner, with no predetermined rules, has yet to be proposed. The existing method that meets all of these criteria except for scalability and non-convexity is the algorithm of K\"ohler \textit{et al.}~\cite{Kohler2018b}.

The main contribution in this paper is thus a DEMPC coordination algorithm that is scalable, fully distributed, and that guarantees stability and convergence in the presence of non-convex cost functions and unknown stationary set-points. In brief, our approach borrows from the method of conflict resolution observed in nature. Namely, when it becomes apparent that agents operating in parallel generate conflicting decisions, a social hierarchy is established to yield resolution. Additionally, proofs of convergence, recursive feasibility, and bounded closed-loop stability are provided along with validation using a numerical example.

\subsection{Paper organization}

The remainder of this paper is organized as follows: Section~\ref{Section - Problem description} provides a description of the nonlinear systems and cost functions that the proposed algorithm addresses, along with an explanation of conflict and convergence issues arising from non-convex objectives; Section~\ref{Section - DEMPC algorithm} highlights the proposed DEMPC algorithm along with proofs of convergence, feasibility, and stability; Section~\ref{Section - Numerical example} implements the proposed method on a numerical example with non-convex cost functions; and finally, Section~\ref{Section - DEMPC - Conclusions} concludes the paper with a summary of major findings, along with recommendations for future research directions.

\section{Problem description} \label{Section - Problem description}


\subsection{Notation}

This brief subsection introduces the reader to the set theory and other notation used in this work. The term $\mathbb{I}_{a:b}$ indicates a set of real integers ranging from $a$ to $b$. The symbols $\mathbf{x} \in \mathbb{R}^n$ state that $\mathbf{x}$ is a real-valued vector of dimensions $n \times 1$. The expression $\mathcal{A} \setminus \mathcal{B}$ denotes the difference between the sets $\mathcal{A}$ and $\mathcal{B}$ (\textit{i.e.} the set $\mathcal{A}$ with all elements of set $\mathcal{B}$ removed). The operation $\mathcal{A} \times \mathcal{B}$ yields the Cartesian product of the sets $\mathcal{A}$ and $\mathcal{B}$.

\subsection{Dynamic model}

We consider $N$ agents that are dynamically decoupled and uninfluenced by disturbances. The dynamics of each agent~$i \in \mathcal{I} = \left\{ 1, 2, \cdots, N \right\}$ are represented by the following discrete-time nonlinear state-space model:
\begin{equation}
	\mathbf{x}_i^+ = \mathbf{f}_i( \mathbf{x}_i, \mathbf{u}_i),
\end{equation}
where $\mathbf{x}_i \in \mathbb{R}^{n_i}$ and $\mathbf{u}_i \in \mathbb{R}^{m_i}$ denote vectors containing the $n_i$ states and $m_i$ inputs of agent~$i$, and $\mathbf{x}_i^+$ represents $\mathbf{x}_i$ at the subsequent sampling time-step. We consider the case where $\mathbf{x}_i$ and $\mathbf{u}_i$ must be bounded within the convex sets $\mathcal{X}_i$ and $\mathcal{U}_i$ at all times, which results in the following state and input constraints:
\begin{eqnarray}
	\mathbf{u}_i & \in & \mathcal{U}_i, \\
	\mathbf{x}_i & \in & \mathcal{X}_i.
\end{eqnarray}

With these operational bounds defined, we make the following assumptions concerning controllability and continuity.
\begin{assumption} \label{Assumption - Controllability}
	(Weak controllability) Let the set $\mathcal{Z}_i^s$ comprise all feasible stationary points of agent~$i$ as follows:
	\begin{equation}
		\mathcal{Z}_i^s \coloneqq \left\{ (\mathbf{x}_i, \mathbf{u}_i) \in \mathcal{X}_i \times \mathcal{U}_i~|~\mathbf{x}_i = \mathbf{f}_i(\mathbf{x}_i, \mathbf{u}_i) \right\}.
	\end{equation}
	All feasible stationary state vectors of agent~$i$ may then be collected within the set $\mathcal{X}_i^s$, which is defined as follows:
	\begin{equation}
		\mathcal{X}_i^s \coloneqq \left\{ \mathbf{x}_i \in \mathcal{X}_i~|~\exists \mathbf{u}_i \in \mathcal{U}_i : \left( \mathbf{x}_i, \mathbf{u}_i \right) \in \mathcal{Z}_i^s \right\}.
	\end{equation}
	
	Let the set $\mathcal{Z}_i^{0 \rightarrow s}$ contain all pairings of initial state vectors $\mathbf{x}_i^0$ and input trajectories $\overline{\mathbf{u}}_i = \left( \mathbf{u}_i^0, \mathbf{u}_i^1, \cdots, \mathbf{u}_i^{H - 1} \right)$ that steer agent~$i$ to each feasible stationary point $\mathbf{x}_i^s$ in $H$ time-steps, while satisfying constraints. $\mathcal{Z}_i^{0 \rightarrow s}$ may therefore be defined as follows:
	\begin{equation}
		\begin{split}
			\mathcal{Z}_i^{0 \rightarrow s} \coloneqq & \left\{ (\mathbf{x}_i^0, \overline{\mathbf{u}}_i, \mathbf{x}_i^s) \in \mathcal{X}_i \times \overline{\mathcal{U}}_i \times \mathcal{X}_i^s~| \right. \\
			& \left. \exists \mathbf{x}_i^1, \mathbf{x}_i^2, \cdots, \mathbf{x}_i^H : \mathbf{x}_i^k = \mathbf{f}_i(\mathbf{x}_i^{k - 1}, \mathbf{u}_i^{k - 1}), \right. \\
			& \left. \mathbf{x}_i^k \in \mathcal{X}_i, \forall k \in \mathbb{I}_{1 : H}, \mathbf{x}_i^H = \mathbf{x}_i^s \right\},
		\end{split}
	\end{equation}
	where $\overline{\mathcal{U}}_i = \mathcal{U}_i \times \cdots \times \mathcal{U}_i = \mathcal{U}_i^H$. All possible initial state vectors $\mathbf{x}_i^0$ that may be steered to a feasible stationary point $\mathbf{x}_i^s$, with constraint satisfaction, are then contained within the set $\mathcal{X}_i^{0 \rightarrow s}$ defined as follows:
	\begin{equation} \label{Equation - Feasible initial state vector set}
		\mathcal{X}_i^{0 \rightarrow s} \coloneqq \left\{ \mathbf{x}_i^0 \in \mathcal{X}_i~|~\exists \overline{\mathbf{u}}_i \in \overline{\mathcal{U}}_i, \mathbf{x}_i^s \in \mathcal{X}_i^s : \left( \mathbf{x}_i^0, \overline{\mathbf{u}}_i, \mathbf{x}_i^s \right) \in \mathcal{Z}_i^{0 \rightarrow s} \right\}.
	\end{equation}
	
	For any agent~$i \in \mathcal{I}$, any initial state vector $\mathbf{x}_i^0 \in \mathcal{X}_i^{0 \rightarrow s}$, input vector trajectory $\overline{\mathbf{u}}_i \in \overline{\mathcal{U}}_i$, and stationary state vector $\mathbf{x}_i^s \in \mathcal{X}_i^s$ such that $\left( \mathbf{x}_i^0, \overline{\mathbf{u}}_i, \mathbf{x}_i^s \right) \in \mathcal{Z}_i^{0 \rightarrow s}$, and any stationary input vector $\mathbf{u}_i^s \in \mathcal{U}_i$ such that $\left( \mathbf{x}_i^s, \mathbf{u}_i^s \right) \in \mathcal{Z}_i^s$, there exists a $\mathcal{K}_\infty$ function $\gamma(\cdot)$ that satisfies the following condition:
	\begin{equation}
		\sum_{k = 0}^{H - 1} \Vert \mathbf{u}_i^k - \mathbf{u}_i^s \Vert \leq \gamma(\Vert \mathbf{x}_i^0 - \mathbf{x}_i^s \Vert).
	\end{equation}
\end{assumption}
\begin{remark}
	The weak controllability assumption simply states that, for any feasible stationary point, there exists some surrounding set from which an initial state vector may be steered to the stationary point. This assumption is necessary for guaranteeing feasibility of the optimization problem of the DEMPC algorithm since reaching a stationary target is one of its constraints. Therefore, if an input trajectory exists that can steer an initial state vector to a stationary point, then a solution to the optimization problem exists that satisfies its constraints.
\end{remark}

\begin{assumption} \label{Assumption - Lipschitz dynamics}
	(Lipschitz continuous dynamics) For any agent~$i \in \mathcal{I}$, $\mathbf{f}_i(\cdot)$ satisfies the following condition for Lipschitz continuity for all $(\mathbf{x}_i^a, \mathbf{u}_i^a), (\mathbf{x}_i^b, \mathbf{u}_i^b) \in \mathcal{X}_i \times \mathcal{U}_i$:
	\begin{equation}
		\left\Vert \mathbf{f}_i(\mathbf{x}_i^b, \mathbf{u}_i^b) - \mathbf{f}_i(\mathbf{x}_i^a, \mathbf{u}_i^a) \right\Vert \leq \Lambda_i^f \left\Vert (\mathbf{x}_i^b, \mathbf{u}_i^b) - (\mathbf{x}_i^a, \mathbf{u}_i^a) \right\Vert,
	\end{equation}
	where the scalar $\Lambda_i^f \geq 0$ is the Lipschitz constant of $\mathbf{f}_i(\cdot)$ on the set $\mathcal{X}_i \times \mathcal{U}_i$.
\end{assumption}
\begin{remark}
	Lipschitz continuity simply states that the function $\mathbf{f}_i(\cdot)$ must be continuous. In other words, there must exist no discontinuities along $\mathbf{f}_i(\cdot)$ that lead to an undefined gradient. This assumption is necessary for guaranteeing optimality in the solution of an optimization problem. If gradients are defined, a local minimum of a cost function will be reached after a sufficient number of iterations.
\end{remark}

\subsection{Control objective}

At each time-step, the control objective of agent~$i$ is to minimize a cooperative economic stage cost function $J_i(\cdot)$ over a finite prediction horizon $H$ as follows:
\begin{equation}
	\min \sum_{k = 0}^{H - 1} J_i(\mathbf{x}_i^k, \mathbf{u}_i^k, \mathbf{x}_{-i|J}^k, \mathbf{u}_{-i|J}^k),
\end{equation}
where the superscript $k$ identifies the time-step number along the prediction horizon $H$, $\mathbf{x}_i^k$ and $\mathbf{u}_i^k$ denote the state and input vectors of agent~$i$ at time-step $k$ along the prediction horizon, and $\mathbf{x}_{-i|J}^k$ and $\mathbf{u}_{-i|J}^k$ contain the state and input vectors at time-step $k$ along the prediction horizon of all agents~$j \in \mathcal{I} \setminus i$ that influence the cooperative cost function $J_i(\cdot)$ of agent~$i$. We collect the indices of these agents into the set $\mathcal{N}_{-i|J}$. Likewise, the indices of all agents~$j \in \mathcal{I} \setminus i$ whose cooperative stage cost functions $J_j(\cdot)$ are influenced by $\mathbf{x}_i$ and $\mathbf{u}_i$ are collected into the set $\mathcal{N}_{+i|J}$.

The objective function $J_i(\cdot)$ may be non-convex; however, it must adhere to the following assumptions concerning cooperation, boundedness, and continuity.
\begin{assumption} \label{Assumption - Neighborhood cooperative objectives}
	(Neighborhood-cooperative objectives) Let each agent~$i \in \mathcal{I}$ possess a stage cost function $\ell_i(\cdot)$ that represents its local economic interests. Then, let the set $\mathcal{N}_{-i}$ contain the indices of all agents~$j \in \mathcal{I} \setminus i$ whose state and input vectors $\mathbf{x}_j$ and $\mathbf{u}_j$ influence the local stage cost function $\ell_i(\cdot)$. Likewise, let the set $\mathcal{N}_{+i}$ contain the indices of all agents~$j \in \mathcal{I} \setminus i$ whose local stage cost functions $\ell_j(\cdot)$ are influenced by $\mathbf{x}_i$ and $\mathbf{u}_i$.
	
	The stage cost function $J_i(\cdot)$ for any agent~$i \in \mathcal{I}$ is neighborhood-cooperative in that it comprises the local interests of agent~$i$ and those of each downstream neighbor $j \in \mathcal{N}_{+i}$ as follows:
	\begin{equation} \label{Equation - Cooperative cost function}
		J_i(\mathbf{x}_i, \mathbf{u}_i, \mathbf{x}_{-i|J}, \mathbf{u}_{-i|J}) \coloneqq \ell_i(\mathbf{x}_i, \mathbf{u}_i, \mathbf{x}_{-i}, \mathbf{u}_{-i}) + \sum_{j \in \mathcal{N}_{+i}} \ell_j(\mathbf{x}_j, \mathbf{u}_j, \mathbf{x}_{-j}, \mathbf{u}_{-j}).
	\end{equation}
	The vectors $\mathbf{x}_{-i}$ and $\mathbf{u}_{-i}$ contain the states and inputs of all agents~$j \in \mathcal{N}_{-i}$. Note that $\mathcal{N}_{-i|J} = \mathcal{N}_{-i} \cup \mathcal{N}_{+i} \cup \mathcal{N}_{-j} \forall j \in \mathcal{N}_{+i}$.
\end{assumption}
\begin{remark}
	The purpose of the social hierarchy-based DEMPC algorithm presented in this work is to enable coupled agents to compute optimal decisions that are mutually beneficial. A fundamental requirement for this algorithm is therefore that coupled agents share interests; hence the assumption of neighborhood-cooperative cost functions. If this assumption is not present, then the agents are competitive, and mutually beneficial decisions cannot be guaranteed.
\end{remark}

\begin{assumption} \label{Assumption - Bounded cost function minima}
	(Bounded cost function minima) Let the sets $\mathcal{X}_{-i}$ and $\mathcal{U}_{-i}$ be defined as follows:
	\begin{eqnarray}
		\mathcal{X}_{-i} & \coloneqq & \prod_{j \in \mathcal{N}_{-i}} \mathcal{X}_j, \\
		\mathcal{U}_{-i} & \coloneqq & \prod_{j \in \mathcal{N}_{-i}} \mathcal{U}_j,
	\end{eqnarray}
	For any agent~$i \in \mathcal{I}$, there exist state and input vectors $\left( \mathbf{x}_i^*, \mathbf{u}_i^*, \mathbf{x}_{-i}^*, \mathbf{u}_{-i}^* \right) \in \mathcal{X}_i \times \mathcal{U}_i \times \mathcal{X}_{-i} \times \mathcal{U}_{-i}$ such that the following condition holds for all $\left( \mathbf{x}_i, \mathbf{u}_i, \mathbf{x}_{-i}, \mathbf{u}_{-i} \right) \in \mathcal{X}_i \times \mathcal{U}_i \times \mathcal{X}_{-i} \times \mathcal{U}_{-i}$:
	\begin{equation}
		\ell_i(\mathbf{x}_i^*, \mathbf{u}_i^*, \mathbf{x}_{-i}^*, \mathbf{u}_{-i}^*) \leq \ell_i(\mathbf{x}_i, \mathbf{u}_i, \mathbf{x}_{-i}, \mathbf{u}_{-i}).
	\end{equation}
\end{assumption}
\begin{remark}
	The bounded cost function minima assumption is necessary for guaranteeing convergence of optimization problems. If the global minimum of a cost function is finite, then an optimization algorithm that descends along the gradients of the cost function is guaranteed to reach a point that satisfies optimality conditions after a sufficient number of iterations.
\end{remark}

\begin{assumption} \label{Assumption - Lipschitz cost functions}
	(Lipschitz continuous objectives) For any agent~$i \in \mathcal{I}$, the local cost function $\ell_i(\cdot)$ satisfies the following condition for Lipschitz continuity for all $(\mathbf{x}_i^a, \mathbf{u}_i^a, \mathbf{x}_{-i}^a, \mathbf{u}_{-i}^a), (\mathbf{x}_i^b, \mathbf{u}_i^b, \mathbf{x}_{-i}^b, \mathbf{u}_{-i}^b) \in \mathcal{X}_i \times \mathcal{U}_i \times \mathcal{X}_{-i} \times \mathcal{U}_{-i}$:
	\begin{equation}
		\left\Vert \ell_i(\mathbf{x}_i^b, \mathbf{u}_i^b, \mathbf{x}_{-i}^b, \mathbf{u}_{-i}^b) - \ell_i(\mathbf{x}_i^a, \mathbf{u}_i^a, \mathbf{x}_{-i}^a, \mathbf{u}_{-i}^a) \right\Vert \leq \Lambda_i^\ell \left\Vert (\mathbf{x}_i^b, \mathbf{u}_i^b, \mathbf{x}_{-i}^b, \mathbf{u}_{-i}^b) - (\mathbf{x}_i^a, \mathbf{u}_i^a, \mathbf{x}_{-i}^a, \mathbf{u}_{-i}^a) \right\Vert,
	\end{equation}
	where the scalar $\Lambda_i^\ell \geq 0$ is the Lipschitz constant of $\ell_i(\cdot)$ on the set $\mathcal{X}_i \times \mathcal{U}_i \times \mathcal{X}_{-i} \times \mathcal{U}_{-i}$.
\end{assumption}
\begin{remark}
	Refer to the remark of Assumption~\ref{Assumption - Lipschitz dynamics} for a simplified explanation of Lipschitz continuity.
\end{remark}

\subsection{Conflict under non-convexity} \label{Subsection - Conflict explanation}

In this subsection, we elaborate further on the main challenge that is associated with non-convex cost functions in DEMPC. Consider a simple problem with only two optimization variables $z_1$ and $z_2$, which are computed by agents~1 and~2, respectively. Further, let both agents share a common non-convex global objective function with contours plotted in Fig.~\ref{Figure - Example non-convex optimization}.

\begin{figure}
	\centering
	\includegraphics[width=3.5in]{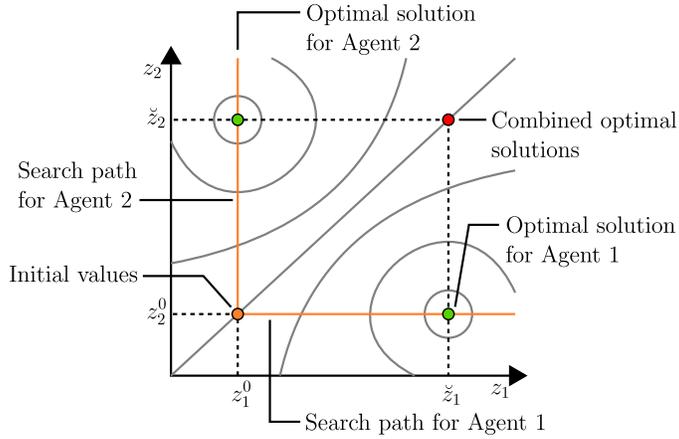}
	\caption{A visual explanation of conflict generated as a result of non-convexity in distributed parallel optimization.} \label{Figure - Example non-convex optimization}
\end{figure}

Assume initial values $z_1^0$ and~$z_2^0$ obtained from a previous iteration or time-step. Under parallel and fully distributed operation, each agent must assume that its neighbors optimization variable remains unchanged while locally minimizing the global objective function. As a result, agent~1 assumes that $z_2$ remains fixed at $z_2^0$ and restricts its search path to the horizontal orange line shown in Fig.~\ref{Figure - Example non-convex optimization}. Likewise, agent~2 assumes that $z_1$ is maintained at $z_1^0$, which constrains its search path to the vertical orange line.

Upon completion of its local optimization problem, agent~1 finds the local optimum located at $\left( \breve{z}_1, z_2^0  \right)$. Agent~2 achieves the same at $\left( z_1^0, \breve{z}_2 \right)$. When the updated optimal variables $\breve{z}_1$ and~$\breve{z}_2$ are combined however, the overall system operates at neither of the local optima identified by the individual agents. We refer to such an outcome as conflict in the current work. Specifically, we define conflict and conflict-free operation as follows.
\begin{definition} \label{Definition - Conflict}
	(Conflict) Agent~$i$ encounters conflict when its economic performance deteriorates upon considering the optimal control actions of its neighbors. More formally, consider $\hat{V}_i^s$ and $\breve{V}_i^s$ defined as follows:
	\begin{eqnarray}
		\hat{V}_i^s & \coloneqq & J_i(\breve{\mathbf{x}}_i^s, \breve{\mathbf{u}}_i^s, \hat{\mathbf{x}}_{-i|J}^s, \hat{\mathbf{u}}_{-i|J}^s), \label{Equation - Naive optimal stage cost function} \\
		\breve{V}_i^s & \coloneqq & J_i(\breve{\mathbf{x}}_i^s, \breve{\mathbf{u}}_i^s, \breve{\mathbf{x}}_{-i|J}^s, \breve{\mathbf{u}}_{-i|J}^s), \label{Equation - Informed optimal stage cost function}
	\end{eqnarray}
	where $\breve{\mathbf{x}}_i^s$ and $\breve{\mathbf{u}}_i^s$ denote the optimal stationary state and input vectors computed by agent~$i$, $\hat{\mathbf{x}}_{-i|J}^s$ and $\hat{\mathbf{u}}_{-i|J}^s$ contain stationary state and input vectors that agent~$i$ assumes for all neighbors~$j \in \mathcal{N}_{-i|J}$, and $\breve{\mathbf{x}}_{-i|J}^s$ and $\breve{\mathbf{u}}_{-i|J}^s$ consist of optimal state and input vectors computed by all agents $j \in \mathcal{N}_{-i|J}$ and communicated to agent~$i$. The terms $\hat{V}_i^s$ and $\breve{V}_i^s$ represent naive and informed values of the stage cost function of agent~$i$ at some stationary point, respectively. The term naive indicates that the cost function value is computed based on assumed values of neighboring agents' state and input vectors. Contrarily, the term informed is employed when agent~$i$ considers recently communicated updated optimal state and input vectors. Given these definitions, while attempting to negotiate an optimal stationary point $\left( \breve{\mathbf{x}}_i^s, \breve{\mathbf{u}}_i^s \right)$, agent~$i$ operates in conflict with its neighbors if the following statement is true:
	\begin{equation}
		\breve{V}_i^s > \hat{V}_i^s.
	\end{equation}
	
	Similarly, one may define naive and informed values of cost functions summed along the prediction horizon as follows:
	\begin{eqnarray}
		\hat{V}_i & \coloneqq & \sum_{k = 0}^{H - 1} J_i(\breve{\mathbf{x}}_i^k, \breve{\mathbf{u}}_i^k, \hat{\mathbf{x}}_{-i|J}^k, \hat{\mathbf{u}}_{-i|J}^k), \label{Equation - Naive optimal cost function} \\
		\breve{V}_i & \coloneqq & \sum_{k = 0}^{H - 1} J_i(\breve{\mathbf{x}}_i^k, \breve{\mathbf{u}}_i^k, \breve{\mathbf{x}}_{-i|J}^k, \breve{\mathbf{u}}_{-i|J}^k), \label{Equation - Informed optimal cost function}
	\end{eqnarray}
	where $\breve{\mathbf{x}}_i^k$ and $\breve{\mathbf{u}}_i^k$ denote optimal state and input vectors computed by agent~$i$ at time-step $k$ along the prediction horizon, $\hat{\mathbf{x}}_{-i|J}^k$ and $\hat{\mathbf{u}}_{-i|J}^k$ contain state and input vectors that agent~$i$ assumes for all neighbors~$j \in \mathcal{N}_{-i|J}$ at some time-step $k$ along the prediction horizon, and $\breve{\mathbf{x}}_{-i|J}^k$ and $\breve{\mathbf{u}}_{-i|J}^k$ consist of optimal state and input vectors computed by all agents $j \in \mathcal{N}_{-i|J}$ at some time-step $k$ along the prediction horizon and communicated to agent~$i$. Given this information, while attempting to negotiate optimal state and input trajectories $\overline{\mathbf{x}}_i^* = \left( \breve{\mathbf{x}}_i^0, \breve{\mathbf{x}}_i^1, \cdots, \breve{\mathbf{x}}_i^H \right)$ and $\overline{\mathbf{u}}_i^* = \left( \breve{\mathbf{u}}_i^0, \breve{\mathbf{u}}_i^1, \cdots, \breve{\mathbf{u}}_i^{H - 1} \right)$, agent~$i$ operates in conflict with its neighbors if the following statement is true:
	\begin{equation}
		\breve{V}_i > \hat{V}_i.
	\end{equation}
\end{definition}

\begin{definition} \label{Definition - Conflict-free operation}
	(Conflict-free operation) Agent~$i$ operates free of conflict when its economic performance improves or remains unchanged upon considering the optimal control actions of its neighbors. More formally, consider once again the values $\hat{V}_i^s$, $\breve{V}_i^s$, $\hat{{V}}_i$, and $\breve{V}_i$ as previously defined. While attempting to negotiate an optimal stationary point $\left( \breve{\mathbf{x}}_i^s, \breve{\mathbf{u}}_i^s \right)$, agent~$i$ operates free of conflict with its neighbors if the following statement is true:
	\begin{equation} \label{Equation - Conflict-free stationary point negotiation}
		\breve{V}_i^s \leq \hat{V}_i^s.
	\end{equation}
	While attempting to negotiate optimal trajectories $\overline{\mathbf{x}}_i^*$ and $\overline{\mathbf{u}}_i^*$, agent~$i$ operates free of conflict with its neighbors if the following statement is true:
	\begin{equation} \label{Equation - Conflict-free trajectory negotiation}
		\breve{V}_i \leq \hat{V}_i.
	\end{equation}
\end{definition}

In the introduction, several algorithms based on sequential operation, agent negotiation, and agent grouping that could resolve non-convex conflict were discussed. The main disadvantages of these algorithms were lack of scalability, idleness of certain agents, and the requirement of predefined rules. In the next section, we propose a fully distributed and scalable solution to the problem of conflict that addresses these drawbacks by establishing social hierarchies.

\section{Social hierarchy-based DEMPC algorithm} \label{Section - DEMPC algorithm}

\subsection{Social hierarchy framework} \label{Subsection - Social hierarchy framework}

In the current context, a social hierarchy consists of a finite number of levels that establish the sequence in which agents generate decisions in order to resolve conflict. The concept is based loosely on social hierarchies that appear naturally among living organisms as a means to resolve conflict and establish which individuals' decisions take priority over those of others. These hierarchies are often determined by the evolutionary or cultural characteristics of the individuals, which are, in essence, randomly assigned. In a similar fashion, we propose a framework which permits the formation of hierarchies with elements of randomness in order to resolve conflict resulting from non-convexity.

For generality, we assume an iterative parallel coordination algorithm, the first of which was presented by Du \textit{et al.}~\cite{Du2001}. Agents synchronously solve their local optimization problems and communicate repeatedly within a single sampling time-step until some termination condition is satisfied or until a maximum number of iterations have been implemented. The socially hierarchy framework is equally applicable to non-iterative parallel methods which were first investigated by Jia and Krogh~\cite{Jia2001} and Dunbar and Murray~\cite{Dunbar2006}.

A visual representation of the social hierarchy framework is shown in Fig.~\ref{Figure - Pecking order framework}. Within a single iteration, there exist $N_q$ hierarchy levels which specify the order in which agents make decisions. Each agent may solve for its optimal stationary point and control trajectory only once within an iteration; however, this computation may take place within any hierarchy level. During each iteration, agents occupying hierarchy level $q = 1$ make decisions first and transfer relevant information to their neighbors. Following this step, agents allocated to hierarchy level $q = 2$ perform the same task. This trend continues until all hierarchy level computations have been performed, at which point the entire process is repeated during the next iteration. It is important to note that multiple agents may occupy the same hierarchy level, and that $N_q$ may be substantially smaller than $N$.

Two fundamental questions now arise regarding (i) how agents should sort themselves among the $N_q$ hierarchy levels in order to resolve conflict, and (ii) how should $N_q$ be determined by the control system designer. The former concern is addressed in Section~\ref{Subsection - DEMPC algorithm}, which describes and assesses a novel DEMPC coordination algorithm that utilizes the concept of a social hierarchy. The latter question is discussed in Section~\ref{Subsection - Social hierarchy properties}, which burrows elements from vertex coloring theory to establish social hierarchy properties.

\begin{figure}
	\centering
	\includegraphics[width=2.25in]{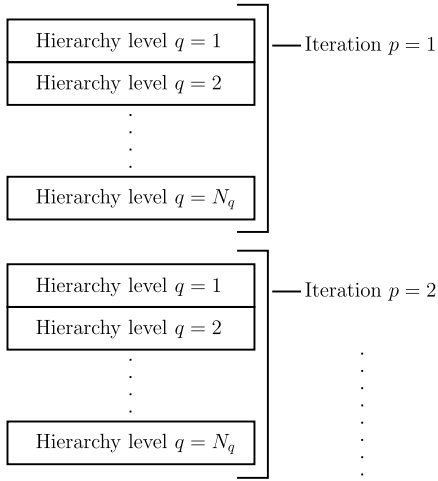}
	\caption{Schematic demonstrating the organization structure of the social hierarchy framework.} \label{Figure - Pecking order framework}
\end{figure}

\subsection{DEMPC coordination algorithm} \label{Subsection - DEMPC algorithm}

This subsection details a novel coordination algorithm for DEMPC with non-convex objectives based on the social hierarchy framework described in Section~\ref{Subsection - Social hierarchy framework}. Our approach allows agents to resolve their conflicts in a truly distributed and scalable manner without the requirement of predefined rules or access to the full system interaction topology. Each agent achieves this outcome by occupying an appropriate level along a social hierarchy when conflict arises. In brief, our algorithm follows evolutionary principles. That is, if an agent occupying a particular hierarchy level during some iteration experiences conflict, then its current hierarchy level is detrimental to its performance and must be randomly mutated.

While negotiating optimal trajectories $\overline{\mathbf{x}}_i^*$ and $\overline{\mathbf{u}}_i^*$, agent~$i$ solves the following optimization problem during each iteration:
\begin{equation} \label{Equation - DEMPC - Dynamic optimization problem}
	\min_{\overline{\mathbf{u}}_i, \overline{\mathbf{x}}_i} \sum_{k = 0}^{H - 1} J_i(\mathbf{x}_i^k, \mathbf{u}_i^k, \hat{\mathbf{x}}_{-i|J}^k, \hat{\mathbf{u}}_{-i|J}^k),
\end{equation}
subject to
\begin{subequations} \label{Equation - DEMPC - Dynamic optimization problem constraints}
	\begin{eqnarray}
		\mathbf{x}_i^0 & = & \mathbf{x}_i, \label{Equation - State initial condition} \\
		\mathbf{x}_i^{k + 1} & = & \mathbf{f}_i(\mathbf{x}_i^k, \mathbf{u}_i^k), \label{Equation - Dynamics constraint} \\
		\mathbf{x}_i^k & \in & \mathcal{X}_i, \label{Equation - State vector constraint} \\ 
		\mathbf{u}_i^k & \in & \mathcal{U}_i, \label{Equation - Input vector constraint} \\
		\mathbf{x}_i^H & = & \mathbf{x}_i^s, \label{Equation - Terminal state vector constraint}
	\end{eqnarray}
\end{subequations}
where $\mathbf{x}_i^k$ and $\mathbf{u}_i^k$ denote the candidate state and input vectors of agent~$i$ at some time-step $k$ along the prediction horizon $H$, $\overline{\mathbf{x}}_i$ and $\overline{\mathbf{u}}_i$ represent candidate state and input vector trajectories as follows:
\begin{eqnarray}
	\overline{\mathbf{x}}_i & \coloneqq & \left( \mathbf{x}_i^0, \mathbf{x}_i^1, \cdots, \mathbf{x}_i^H \right),\\
	\overline{\mathbf{u}}_i & \coloneqq & \left( \mathbf{u}_i^0, \mathbf{u}_i^1, \cdots, \mathbf{u}_i^{H - 1} \right),
\end{eqnarray}
and $\hat{\mathbf{x}}_{-i|J}^k$ and $\hat{\mathbf{u}}_{-i|J}^k$ contain the state and input vectors that agent~$i$ assumes for all neighbors~$j \in \mathcal{N}_{-i|J}$ at time-step $k$. For notational brevity in Algorithm~\ref{Algorithm - DEMPC coordination algorithm}, we condense the assumed states and input of agent~$j$ into the trajectories $\mathbf{\tilde{x}}_j$ and $\mathbf{\tilde{u}}_j$ as follows:
\begin{eqnarray}
	\mathbf{\tilde{x}}_j & \coloneqq & \left( \hat{\mathbf{x}}_j^0, \hat{\mathbf{x}}_j^1, \cdots, \hat{\mathbf{x}}_j^H \right),\\
	\mathbf{\tilde{u}}_j & \coloneqq & \left( \hat{\mathbf{u}}_j^0, \hat{\mathbf{u}}_j^1, \cdots, \hat{\mathbf{u}}_j^{H - 1} \right).
\end{eqnarray}

Solving Problem~(\ref{Equation - DEMPC - Dynamic optimization problem}) yields the optimal state and input vector trajectories $\overline{\mathbf{x}}_i^*$ and $\overline{\mathbf{u}}_i^*$ defined as follows:
\begin{eqnarray}
	\overline{\mathbf{x}}^*_i & \coloneqq & \left( \breve{\mathbf{x}}_i^0, \breve{\mathbf{x}}_i^1, \cdots, \breve{\mathbf{x}}_i^H \right),\\
	\overline{\mathbf{u}}^*_i & \coloneqq & \left( \breve{\mathbf{u}}_i^0, \breve{\mathbf{u}}_i^1, \cdots, \breve{\mathbf{u}}_i^{H - 1} \right),
\end{eqnarray}
where $\breve{\mathbf{x}}_i^k$ and $\breve{\mathbf{u}}_i^k$ denote the optimal state and input vectors at time-step $k$ along the prediction horizon. The difference between \textit{candidate} and \textit{optimal} solutions is that the optimal solutions $\overline{\mathbf{x}}_i^*$ and $\overline{\mathbf{u}}_i^*$ are not immediately \textit{accepted} by the DEMPC algorithm. In other words, if it was required that DEMPC algorithm implement a decision immediately, the enacted action would follow the candidate trajectories $\overline{\mathbf{x}}_i$ and $\overline{\mathbf{u}}_i$, not the optimal solutions $\overline{\mathbf{x}}_i^*$ and $\overline{\mathbf{u}}_i^*$. As detailed in Algorithm~\ref{Algorithm - DEMPC coordination algorithm}, $\overline{\mathbf{x}}^*_i$ and $\overline{\mathbf{u}}^*_i$ are only accepted as candidate solutions if they do not yield conflict. If $\overline{\mathbf{x}}^*_i$ and $\overline{\mathbf{u}}^*_i$ do yield conflict, then the previous values of $\overline{\mathbf{x}}_i$ and $\overline{\mathbf{u}}_i$ remain as the candidate trajectories to be implemented.

Constraint~(\ref{Equation - State initial condition}) serves as the initial condition of the prediction horizon by setting the candidate state vector of agent~$i$ at $k = 0$ equal to the most recent state measurement $\mathbf{x}_i$. Constraint~(\ref{Equation - Dynamics constraint}) requires that the optimal input and state trajectories $\overline{\mathbf{x}}^*_i$ and $\overline{\mathbf{u}}^*_i$ computed along the prediction horizon satisfy the plant dynamics of agent~$i$. Constraints~(\ref{Equation - State vector constraint}) and~(\ref{Equation - Input vector constraint}) state that the optimal input and state vectors $\breve{\mathbf{x}}_i^k$ and $\breve{\mathbf{u}}_i^k$ of agent~$i$ must remain within the process constraint sets $\mathcal{X}_i$ and $\mathcal{U}_i$, respectively, at any time-step $k$ along the prediction horizon. Finally, Constraint~(\ref{Equation - Terminal state vector constraint}) ensures that, by the end of the prediction horizon $H$, the computed optimal input trajectory $\overline{\mathbf{u}}_i^*$ leads the local state vector to a feasible candidate steady state $\mathbf{x}_i^s$.

To negotiate the candidate stationary state vector $\mathbf{x}_i^s$, agent~$i$ solves the following optimization problem during each iteration:
\begin{equation} \label{Equation - DEMPC - Stationary optimization problem}
	\min_{\overline{\mathbf{u}}_i, \mathbf{u}_i^s, \overline{\mathbf{x}}_i, \mathbf{x}_i^s} J_i(\mathbf{x}_i^s, \mathbf{u}_i^s, \hat{\mathbf{x}}_{-i|J}^s, \hat{\mathbf{u}}_{-i|J}^s),
\end{equation}
subject to
\begin{subequations} \label{Equation - DEMPC - Stationary optimization problem constraints}
	\begin{eqnarray}
		\mathbf{x}_i^0 & = & \mathbf{x}_i, \\
		\mathbf{x}_i^{k + 1} & = & \mathbf{f}_i(\mathbf{x}_i^k, \mathbf{u}_i^k), \forall k \in \mathbb{I}_{0:H - 1}, \\
		\mathbf{x}_i^s & = & \mathbf{f}_i(\mathbf{x}_i^s, \mathbf{u}_i^s), \label{Equation - Final stationary point constraint} \\
		\mathbf{x}_i^s & = & \mathbf{x}_i^{H}, \label{Equation - Final equality} \\
		\mathbf{x}_i^k & \in & \mathcal{X}_i, \forall k \in \mathbb{I}_{0:H}, \\ 
		\mathbf{u}_i^k & \in & \mathcal{U}_i, \forall k \in \mathbb{I}_{0:H - 1}, \\
		\left( \mathbf{x}_i^s, \mathbf{u}_i^s \right) & \in & \mathcal{X}_i \times \mathcal{U}_i, \label{Equation - DEMPC - Stationary point process constraint}
	\end{eqnarray}
\end{subequations}
where $\mathbf{x}_i^s$ and $\mathbf{u}_i^s$ denote the candidate stationary state and input vectors of agent~$i$, and $\hat{\mathbf{x}}_{-i|J}^s$ and $\hat{\mathbf{u}}_{-i|J}^s$ contain similar information that agent~$i$ assumes for all neighbors~$j \in \mathcal{N}_{-i|J}$. The steady state and input vectors that agent~$i$ assumes for an individual neighbor~$j$ are denoted by $\hat{\mathbf{x}}_j^s$ and $\hat{\mathbf{u}}_j^s$, respectively. Solving Problem~(\ref{Equation - DEMPC - Stationary optimization problem}) yields the optimal steady state and input vectors $\breve{\mathbf{x}}_i^s$ and $\breve{\mathbf{u}}_i^s$. Similar to the optimal trajectories $\overline{\mathbf{x}}^*_i$ and $\overline{\mathbf{u}}^*_i$, the optimal stationary solution is only accepted if it yields conflict-free operation. Otherwise, the previous values of $\mathbf{x}_i^s$ and $\mathbf{u}_i^s$ remain as the candidate stationary terminal set-point.

Constraint~(\ref{Equation - Final stationary point constraint}) ensures that the computed optimum $\left( \breve{\mathbf{x}}_i^s, \breve{\mathbf{u}}_i^s \right)$ is a stationary point. Constraint~(\ref{Equation - Final equality}) then requires that the computed state vector trajectory ends at the optimal steady state $\breve{\mathbf{x}}_i^s$. Finally, Constraint~(\ref{Equation - DEMPC - Stationary point process constraint}) states that the computed optimal stationary point $\left( \breve{\mathbf{x}}_i^s, \breve{\mathbf{u}}_i^s \right)$ must lie within the process constraint set $\mathcal{X}_i \times \mathcal{U}_i$.

The difference between Problems~(\ref{Equation - DEMPC - Dynamic optimization problem}) and~(\ref{Equation - DEMPC - Stationary optimization problem}) is that the latter only considers the stage cost function $J_i(\cdot)$ and therefore computes a feasible optimal steady state vector $\breve{\mathbf{x}}_i^s$ without minimizing $J_i(\cdot)$ over the prediction horizon. The reason that dynamics are considered in Problem~(\ref{Equation - DEMPC - Stationary optimization problem}) is to ensure that the computed optimal steady state vector $\breve{\mathbf{x}}_i^s$ is reachable from the initial state vector $\mathbf{x}_i^0$. Problem~(\ref{Equation - DEMPC - Dynamic optimization problem}) then computes optimal trajectories that minimize $J_i(\cdot)$ over the prediction horizon and steer the system to the reachable candidate steady state vector $\mathbf{x}_i^s$. A description of the proposed social hierarchy-based DEMPC algorithm now follows.

\begin{algorithm} \label{Algorithm - DEMPC coordination algorithm}
	Social hierarchy-based DEMPC coordination scheme. Implement in parallel for all agents $i \in \mathcal{I}$.
\end{algorithm}
\noindent \textit{Communication protocol}:
\begin{itemize}[leftmargin=*]
	\item Send $\overline{\mathbf{x}}_i$, $\overline{\mathbf{u}}_i$, $\mathbf{x}_i^s$, and $\mathbf{u}_i^s$ to all agents~$j \in \mathcal{N}_{+i|J}$, receive $\overline{\mathbf{x}}_j$, $\overline{\mathbf{u}}_j$, $\mathbf{x}_j^s$, and $\mathbf{u}_j^s$ from all agents~$j \in \mathcal{N}_{-i|J}$, and set $\tilde{\mathbf{x}}_j = \overline{\mathbf{x}}_j$, $\tilde{\mathbf{u}}_j = \overline{\mathbf{u}}_j$, $\hat{\mathbf{x}}_j^s = \mathbf{x}_j^s$, and $\hat{\mathbf{u}}_j^s = \mathbf{u}_j^s$ for all $j \in \mathcal{N}_{-i|J}$.
\end{itemize}
\noindent \textit{Initialization}:
\begin{enumerate}[leftmargin=*]
	\item Specify $N_q \geq 1$ and set $q_i  = 1$.
	\item Initialize $\mathbf{x}_i$, $\overline{\mathbf{x}}_i \in \overline{\mathcal{X}}_i = \mathcal{X}_i \times \cdots \times \mathcal{X}_i = \mathcal{X}_i^H$, $\overline{\mathbf{u}}_i \in \overline{\mathcal{U}}_i$, $\mathbf{x}_i^s \in \mathcal{X}_i$, $\mathbf{u}_i^s \in \mathcal{U}_i$ such that $\mathbf{x}_i^0 = \mathbf{x}_i$, $\mathbf{x}_i^{k + 1} = \mathbf{f}_i(\mathbf{x}_i^k, \mathbf{u}_i^k) \forall k \in \mathbb{I}_{0:H - 1}$, $\mathbf{x}_i^H = \mathbf{x}_i^s$, and $\mathbf{x}_i^s = \mathbf{f}_i(\mathbf{x}_i^s, \mathbf{u}_i^s)$.
	\item Implement \textit{communication protocol}.
\end{enumerate}
\noindent \textit{Perform at each new time-step}:
\begin{enumerate}[leftmargin=*]
	\item Measure $\mathbf{x}_i$, and compute $\overline{\mathbf{x}}_i$ such that $\mathbf{x}_i^0 = \mathbf{x}_i$ and $\mathbf{x}_i^{k + 1} = \mathbf{f}_i(\mathbf{x}_i^k, \mathbf{u}_i^k) \forall k \in \mathbb{I}_{0:H - 1}$.
	\item Implement \textit{communication protocol}.
	\item For iteration number $p = 1, 2, \cdots, N_p$, do:
	\begin{enumerate}[leftmargin=0.15cm]
		\item For sequence slot number $q = 1, 2, \cdots, N_q$, do:
		\begin{enumerate}[leftmargin=0.15cm]
			\item If $q = q_i$, (i) solve Problem~(\ref{Equation - DEMPC - Stationary optimization problem}) to acquire $\breve{\mathbf{x}}_i^s$ and $\breve{\mathbf{u}}_i^s$, (ii) compute $\hat{V}_i^s$ according to Eq.~(\ref{Equation - Naive optimal stage cost function}), (iv) send $\breve{\mathbf{x}}_i^s$ and $\breve{\mathbf{u}}_i^s$ to all agents~$j \in \mathcal{N}_{+i|J}$, (v) receive $\breve{\mathbf{x}}_j^s$ and $\breve{\mathbf{u}}_j^s$ from all agents~$j \in \mathcal{N}_{-i|J}$, and update $\hat{\mathbf{x}}_j^s = \breve{\mathbf{x}}_j^s$ and $\hat{\mathbf{x}}_j^s = \breve{\mathbf{x}}_j^s$.
			\item Else, receive $\breve{\mathbf{x}}_j^s$ and $\breve{\mathbf{u}}_j^s$ from all agents~$j \in \mathcal{N}_{-i|J}$, and update $\hat{\mathbf{x}}_j^s = \breve{\mathbf{x}}_j^s$ and $\hat{\mathbf{x}}_j^s = \breve{\mathbf{x}}_j^s$.
		\end{enumerate}
		\item Compute $\breve{V}_i^s$ according to Eq.~(\ref{Equation - Informed optimal stage cost function}). If $\breve{V}_i^s > \hat{V}_i^s$, randomly change $q_i$ with uniform probability and set $\breve{\mathbf{x}}_i^s = \mathbf{x}_i^s$ and $\breve{\mathbf{u}}_i^s = \mathbf{u}_i^s$. Else, if $\breve{V}_i^s \leq \hat{V}_i^s$, update $\mathbf{x}_i^s = \breve{\mathbf{x}}_i^s$ and $\mathbf{u}_i^s = \breve{\mathbf{u}}_i^s$.
		\item Implement \textit{communication protocol}.
	\end{enumerate}
	\item For iteration number $p = 1, 2, \cdots, N_p$, do:
	\begin{enumerate}[leftmargin=0.15cm]
		\item For sequence slot number $q = 1, 2, \cdots, N_q$, do:
		\begin{enumerate}[leftmargin=0.15cm]
			\item If $q = q_i$, (i) solve Problem~(\ref{Equation - DEMPC - Dynamic optimization problem}) to acquire $\overline{\mathbf{x}}_i^*$ and $\overline{\mathbf{u}}_i^*$, (ii) compute $\hat{V}_i$ according to Eq.~(\ref{Equation - Informed optimal cost function}), (iv) send $\overline{\mathbf{x}}_i^*$ and $\overline{\mathbf{u}}_i^*$ to all agents~$j \in \mathcal{N}_{+i|J}$, (v) receive $\overline{\mathbf{x}}_j^*$ and $\overline{\mathbf{u}}_j^*$ from all agents~$j \in \mathcal{N}_{-i|J}$, and update $\tilde{\mathbf{x}}_j = \overline{\mathbf{x}}_j^*$ and $\tilde{\mathbf{u}}_j = \overline{\mathbf{u}}_j^*$.
			\item Else, receive $\overline{\mathbf{x}}_j^*$ and $\overline{\mathbf{u}}_j^*$ from all agents~$j \in \mathcal{N}_{-i|J}$, and update $\tilde{\mathbf{x}}_j = \overline{\mathbf{x}}_j^*$ and $\tilde{\mathbf{u}}_j = \overline{\mathbf{u}}_j^*$.
		\end{enumerate}
		\item Compute $\breve{V}_i$ according to Eq.~(\ref{Equation - Informed optimal cost function}). If $\breve{V}_i > \hat{V}_i$, randomly change $q_i$ with uniform probability and set $\overline{\mathbf{x}}_i^* = \overline{\mathbf{x}}_i$ and $\overline{\mathbf{u}}_i^* = \overline{\mathbf{u}}_i$. Else, if $\breve{V}_i \leq \hat{V}_i$, update $\overline{\mathbf{x}}_i = \overline{\mathbf{x}}_i^*$ and $\overline{\mathbf{u}}_i = \overline{\mathbf{u}}_i^*$.
		\item Implement \textit{communication protocol}.
	\end{enumerate}
	\item Apply candidate control input $\mathbf{u}_i^0$ to the system, update $\overline{\mathbf{u}}_i$ such that $\mathbf{u}_i^k = \mathbf{u}_i^{k + 1} \forall k \in \mathbb{I}_{0:H - 2}$ and $\mathbf{u}_i^{H - 1} = \mathbf{u}_i^s$.
\end{enumerate}

During initialization, step~1 requires the control system designer to specify the quantity $N_q$ of hierarchy levels and to allocate agent~$i$ to the first level. As a result, all agents initially solve their local optimization problems in parallel. Step~2 requires that feasible state and input trajectories $\overline{\mathbf{x}}_i$ and $\overline{\mathbf{u}}_i$ that satisfy the constraints of Eq.~(\ref{Equation - DEMPC - Dynamic optimization problem constraints}), and steer the system to a reachable stationary point $\left( \mathbf{x}_i^s, \mathbf{u}_i^s \right)$, be specified for agent~$i$ given the initial state vector $\mathbf{x}_i$. Finally, step~3 involves the exchange of candidate trajectories $\overline{\mathbf{x}}_j$ and $\overline{\mathbf{u}}_j$ and stationary vectors $\mathbf{x}_j^s$ and $\mathbf{u}_j^s$ between agent~$i$ and all of its neighbors~$j \in \mathcal{N}_{-i|J}$. Agent~$i$ then uses this incoming information to establish the assumed values of its neighbors' states and inputs $\tilde{\mathbf{x}}_j$, $\tilde{\mathbf{u}}_j$, $\hat{\mathbf{x}}_j^s$, and $\hat{\mathbf{u}}_j^s$ for all $j \in \mathcal{N}_{-i|J}$. These assumptions remain unchanged until future communication.

Focusing on the recursive portion of the DEMPC algorithm, agent~$i$ first measures its state vector $\mathbf{x}_i$ and updates its local state trajectory $\overline{\mathbf{x}}_i$ using the most up-to-date candidate input sequence $\overline{\mathbf{u}}_i$. In step~2, this updated information is communicated between neighbors. As a result, prior to step~3, agent~$i$ possesses the most up-to-date trajectories and stationary vectors $\overline{\mathbf{x}}_j$,$\overline{\mathbf{u}}_j$, $\mathbf{x}_j^s$, and $\mathbf{u}_j^s$ of neighbors~$j \in \mathcal{N}_{-i|J}$. The assumptions $\tilde{\mathbf{x}}_j$, $\tilde{\mathbf{u}}_j$, $\hat{\mathbf{x}}_j^s$, and $\hat{\mathbf{u}}_j^s$ for all $j \in \mathcal{N}_{-i|J}$ are also updated as a result.

Step~3 initiates an iterative process within the current sampling time-step with the objective of identifying an appropriate candidate stationary point $\left( \mathbf{x}_i^s, \mathbf{u}_i^s \right)$ to serve as a terminal constraint in future steps. In step~3(a), agent~$i$ cycles through all hierarchy levels sequentially and, during its allocated hierarchy level, solves Problem~(\ref{Equation - DEMPC - Stationary optimization problem}) to obtain an optimal stationary point $\left( \breve{\mathbf{x}}_i^s, \breve{\mathbf{u}}_i^s \right)$. These vectors, along with the assumptions $\hat{\mathbf{x}}_j^s$ and $\hat{\mathbf{u}}_j^s$ for all $j \in \mathcal{N}_{-i|J}$, are then used to calculate the naive stage cost function value $\hat{V}_i^s$ to be used for comparison to $\breve{V}_i^s$ later on. Finally, agent~$i$ exchanges optimal stationary vector information with neighboring agents and updates its assumptions $\hat{\mathbf{x}}_j^s$ and $\hat{\mathbf{u}}_j^s$ for all $j \in \mathcal{N}_{-i|J}$. Outside of its hierarchy level, agent~$i$ remains idle and only receives information from neighbors and updates its assumptions.

After all hierarchy levels have been cycled through, agent~$i$ will have received $\breve{\mathbf{x}}_j^s$ and $\breve{\mathbf{u}}_j^s$ from all neighbors~$j \in \mathcal{N}_{-i|J}$. Step~3(b) requires the computation of the informed stage cost function value $\breve{V}_i^s$. If $\breve{V}_i^s > \hat{V}_i^s$, then agent~$i$ is in conflict with its neighbors while attempting to identify an optimal stationary point, and it is necessary for its hierarchy level $q_i$ to be randomly mutated. Further, if conflict is encountered, then agent~$i$ resets it optimal stationary vectors $\breve{\mathbf{x}}_i^s$ and $\breve{\mathbf{u}}_i^s$ to the candidate stationary vectors $\mathbf{x}_i^s$ and $\mathbf{u}_i^s$. This reset essentially erases $\breve{\mathbf{x}}_i^s$ and $\breve{\mathbf{u}}_i^s$ so that conflict-yielding optimal solutions are no longer communicated. If $\breve{V}_i^s \leq \hat{V}_i^s$, then agent~$i$ is operating free of conflict, and its current hierarchy level may be maintained. Further, since the recently computed stationary vectors $\breve{\mathbf{x}}_i^s$ and $\breve{\mathbf{u}}_i^s$ did not yield conflict, they may replace the candidate stationary vectors $\mathbf{x}_i^s$ and $\mathbf{u}_i^s$. As a result, agent~$i$ updates $\mathbf{x}_i^s$ and $\mathbf{u}_i^s$ using $\breve{\mathbf{x}}_i^s$ and $\breve{\mathbf{u}}_i^s$.

Step~3 is terminated once the maximum number of iterations $N_p$ has been reached. Step~4 essentially repeats step~3, except the candidate terminal state vector $\mathbf{x}_i^s$ has now been updated, and the objective is to update candidate state and input trajectories $\overline{\mathbf{x}}_i$ and $\overline{\mathbf{u}}_i$. The maximum number of iterations available for this step is also $N_p$. Step~5 requires that the control input vector corresponding to the first time-step along the prediction horizon be applied to the system. Additionally, the candidate input trajectory $\overline{\mathbf{u}}_i$ for the next sampling time-step is constructed by concatenating it with the candidate stationary input vector $\mathbf{u}_i^s$.

\subsection{Closed-loop properties}

The current subsection establishes closed-loop properties for Algorithm~\ref{Algorithm - DEMPC coordination algorithm}. We first demonstrate that, regardless of the interaction topology of a multi-agent system, all conflicts may be resolved in a finite number of iterations with some probability. We then address convergence, feasibility, and stability.

\begin{theorem} \label{Theorem - Conflict resolution}
	(Conflict resolution) There exist a finite number of iterations after which, with some probability greater than zero, Inequalities~(\ref{Equation - Conflict-free stationary point negotiation}) and~(\ref{Equation - Conflict-free trajectory negotiation}) are guaranteed to be satisfied at each iteration within a single time-step.
\end{theorem}

\begin{proof}
	This proof consists of two parts; (i) it is first necessary to prove that, for any interconnected system, at least one social hierarchy exists that will ensure system-wide conflict resolution; (ii) it is then proved that the probability of agents self-organizing according to such a social hierarchy is greater than zero during any iteration.
	
	\noindent \textbf{Part I:} Consider a multi-agent system wherein all agents possess the same stage cost function $J(\cdot)$ defined as follows:
	\begin{equation}
		J(\mathbf{x},\mathbf{u}) \coloneqq \sum_{i \in \mathcal{I}} \ell_i(\mathbf{x}_i, \mathbf{u}_i, \mathbf{x}_{-i}, \mathbf{u}_{-i}),
	\end{equation}
	where $\mathbf{x} \in \mathbb{R}^{\sum_{i \in \mathcal{I}} n_i}$ and $\mathbf{u} \in \mathbb{R}^{\sum_{i \in \mathcal{I}} m_i}$ contain the $n_i$ states and $m_i$ inputs of all agents~$i \in \mathcal{I}$, and $J(\cdot)$ is global-cooperative in that it considers the local interests $\ell_i(\cdot)$ of all agents~$i \in \mathcal{I}$. The resulting dynamic optimization problem over the prediction horizon $H$ for any agent~$i$ is therefore defined as follows:
	\begin{equation} \label{Equation - Global dynamic optimization problem}
		\min_{\overline{\mathbf{u}}_i, \overline{\mathbf{x}}_i} \sum_{k = 0}^{H - 1} J(\mathbf{x}^k, \mathbf{u}^k),
	\end{equation}
	with constraints similar to those of Problem~(\ref{Equation - DEMPC - Dynamic optimization problem}). The vectors $\mathbf{x}^k$ and $\mathbf{u}^k$ denote system-wide states and inputs at time-step $k$ along the prediction horizon. The stationary optimization problem for any agent~$i$ is defined as follows:
	\begin{equation} \label{Equation - Global static optimization problem}
		\min_{\overline{\mathbf{u}}_i, \mathbf{u}_i^s, \overline{\mathbf{x}}_i, \mathbf{x}_i^s} J(\mathbf{x}^s, \mathbf{u}^s),
	\end{equation}
	with constraints similar to those of Problem~(\ref{Equation - DEMPC - Stationary optimization problem}). The vectors $\mathbf{x}^s$ and $\mathbf{u}^s$ denote system-wide steady states and inputs. We also define $V(\cdot)$ and $V^s(\cdot)$ as follows:
	\begin{eqnarray}
		V(\overline{\mathbf{x}}, \overline{\mathbf{u}}) & \coloneqq & \sum_{k = 0}^{H - 1} J(\mathbf{x}^k, \mathbf{u}^k),\\
		V^s(\mathbf{x}^s, \mathbf{u}^s) & \coloneqq & J(\mathbf{x}^s, \mathbf{u}^s),
	\end{eqnarray}
	where $\overline{\mathbf{x}} = \left( \mathbf{x}^0, \mathbf{x}^1, \cdots, \mathbf{x}^H \right)$ and $\overline{\mathbf{u}} = \left( \mathbf{u}^0, \mathbf{u}^1, \cdots, \mathbf{x}^{H - 1} \right)$ denote system-wide state and input trajectories over the prediction horizon.
	
	As a reference case, let the agents operate in a fully sequential manner such that, at any given iteration, only one agent solves its local optimization problem. Focusing on the negotiation of a stationary point for now, only agent~$i$ updates $\mathbf{x}_i^s$ and $\mathbf{u}_i^s$ at some iteration $p$, then transmits this information to all other agents. Furthermore, let recursive feasibility, which is established independently in Theorem~\ref{Theorem - Recursive feasibility}, be presumed, and let Assumptions~\ref{Assumption - Lipschitz dynamics}, \ref{Assumption - Bounded cost function minima}, and~\ref{Assumption - Lipschitz cost functions} concerning continuity and bounded minima hold. Under these conditions, upon solving Problem~(\ref{Equation - Global static optimization problem}), each agent~$i \in \mathcal{I}$ is guaranteed to shift $\mathbf{x}^s$ and $\mathbf{u}^s$ until the gradient of $J(\cdot)$ projected along the variable space $\left( \mathbf{x}_i^s, \mathbf{u}_i^s \right)$ satisfies optimality conditions. Since this process occurs sequentially across all agents~$i \in \mathcal{I}$, $V^s(\cdot)$ is guaranteed to decrease or remain unchanged after each subsequent agent's update to $\mathbf{x}_i^s$ and $\mathbf{u}_i^s$, which ensures conflict-free operation as per Definition~\ref{Definition - Conflict-free operation}.
	
	We now prove that the above result may be achieved using social hierarchies that are not necessarily fully sequential (\textit{i.e.} a subset of agents may solve their local optimization problems simultaneously) and also assuming neighborhood-cooperative stage cost functions as per Assumption~\ref{Assumption - Neighborhood cooperative objectives}. From the perspective of any agent~$i \in \mathcal{I}$, $J(\cdot)$ may be arranged as follows:
	\begin{equation} \label{Equation - Rearranged stage cost}
		J(\mathbf{x},\mathbf{u}) = \ell_i(\mathbf{x}_i, \mathbf{u}_i, \mathbf{x}_{-i}, \mathbf{u}_{-i}) + \sum_{j \in \mathcal{N}_{+i}} \ell_j(\mathbf{x}_j, \mathbf{u}_j, \mathbf{x}_{-j}, \mathbf{u}_{-j}) + \sum_{\kappa \in \mathcal{I} \setminus \mathcal{N}_{+i}, \kappa \neq i} \ell_\kappa(\mathbf{x}_\kappa, \mathbf{u}_\kappa, \mathbf{x}_{-\kappa}, \mathbf{u}_{-\kappa}),
	\end{equation}
	where the three terms on the right-hand-side represent, from left to right, the local interests of agent~$i$, the local interests of agents whose cost functions are influenced by agent~$i$, and the local interests of all remaining agents whose cost functions are uninfluenced by agent~$i$. Taking the gradient of $J(\cdot)$ along the variable space $\left( \mathbf{x}_i, \mathbf{u}_i \right)$ yields the following:
	\begin{equation}
		\nabla_{\mathbf{x}_i, \mathbf{u}_i}J(\mathbf{x},\mathbf{u}) = \nabla_{\mathbf{x}_i, \mathbf{u}_i}\ell_i(\mathbf{x}_i, \mathbf{u}_i, \mathbf{x}_{-i}, \mathbf{u}_{-i}) + \nabla_{\mathbf{x}_i, \mathbf{u}_i}\sum_{j \in \mathcal{N}_{+i}} \ell_j(\mathbf{x}_j, \mathbf{u}_j, \mathbf{x}_{-j}, \mathbf{u}_{-j}).
	\end{equation}
	Note that gradient of $\ell_\kappa(\cdot)$ for all $\kappa \in \mathcal{I} \setminus \mathcal{N}_{+i}, \kappa \neq i$ is zero along $\left( \mathbf{x}_i, \mathbf{u}_i \right)$ since these expressions have no dependency upon the operation of agent~$i$. The shape of $J(\cdot)$ along the variable space $\left( \mathbf{x}_i, \mathbf{u}_i \right)$ is therefore only influenced by $\mathbf{x}_j$ and $\mathbf{u}_j$ for all $j \in \mathcal{N}_{-i|J}$, since $\mathcal{N}_{-i|J} = \mathcal{N}_{-i} \cup \mathcal{N}_{+i} \cup \mathcal{N}_{-j} \forall j \in \mathcal{N}_{+i}$. This result yields two important consequences.
	
	First, the local optimization problem of agent~$i$ is uninfluenced by the control actions of agents~$\kappa \in \mathcal{I} \setminus \mathcal{N}_{-i|J}, \kappa \neq i$. Agent~$i$ may therefore update its control actions in parallel with agents~$\kappa \in \mathcal{I} \setminus \mathcal{N}_{-i|J}, \kappa \neq i$, and the computed optimal trajectories are guaranteed to decrease or preserve $V^s(\cdot)$. Second, whether agent~$i$ employs the global stage cost function $J(\cdot)$, or the neighborhood-cooperative stage cost function $J_i(\cdot)$ defined in Eq.~(\ref{Equation - Cooperative cost function}), the computed optimal trajectories remain unchanged. This property is true since $\nabla_{\mathbf{x}_i, \mathbf{u}_i}J(\cdot) = \nabla_{\mathbf{x}_i, \mathbf{u}_i}J_i(\cdot)$. Therefore, if Assumption~\ref{Assumption - Neighborhood cooperative objectives} concerning neighborhood-cooperative cost functions holds, and if no two agents~$i$ and~$j$ such that $j \in \mathcal{N}_{-i|J}$ for all $i, j \in \mathcal{I}, j \neq i$ solve their local optimization problems simultaneously, then at least one social hierarchy exists that will satisfy Inequality~(\ref{Equation - Conflict-free stationary point negotiation}) after each iteration. The above logic may be extended to step~4 in Algorithm~\ref{Algorithm - DEMPC coordination algorithm} without modification. Part~I of the proof is thus completed.
	
	\noindent \textbf{Part II:} Due to its distributed nature, Algorithm~\ref{Algorithm - DEMPC coordination algorithm} may lead some agents to resolve their conflicts earlier than others. However, we prove that, even in a worst-case scenario in which all agents initially encounter conflict, the probability that all conflicts will be resolved within a single iteration is greater than zero. Let $N_s$ describe the quantity of possible social hierarchies that will resolve conflict among $N$ agents. If agents must, with uniform probability, randomly choose among $N_q$ hierarchy levels, the probability $P$ that all conflicts are resolved within a single iteration is defined as follows:
	\begin{equation}
		P = N_s \left( \frac{1}{N_q} \right)^{N}.
	\end{equation}
	Thus, as long as $N_s > 0$, $P > 0$ must also be true. Part~I of this proof demonstrated that $N_s \geq 1$ for any interconnected system with neighborhood-cooperative cost functions.
\end{proof}

\begin{theorem} \label{Theorem - Convergence}
	(Convergence) Let $\mathbf{x}_i^{s,p}$ and $\overline{\mathbf{u}}_i^p$ denote the candidate steady state vector and input trajectory $\mathbf{x}_i^s$ and $\overline{\mathbf{u}}_i$ held by agent~$i$ at iteration $p$. Given a sufficiently large number of iterations in Steps~3 and~4 in Algorithm~\ref{Algorithm - DEMPC coordination algorithm}, the following inequalities are guaranteed to be satisfied for all $i \in \mathcal{I}$:
	\begin{eqnarray}
		\left\Vert \mathbf{x}_i^{s,p} - \mathbf{x}_i^{s,p - 1} \right\Vert & \leq & \alpha \left\Vert \mathbf{x}_i^{s,p - 1} \right\Vert, \label{Equation - Stationary state vector convergence} \\
		\left\Vert \overline{\mathbf{u}}_i^p - \overline{\mathbf{u}}_i^{p - 1} \right\Vert & \leq & \beta \left\Vert \overline{\mathbf{u}}_i^{p - 1} \right\Vert, \label{Equation - Input trajectory convergence}
	\end{eqnarray}
	where $\alpha > 0$ and $\beta > 0$ represent fixed convergence tolerances.
\end{theorem}

\begin{proof}
	Theorem~\ref{Theorem - Conflict resolution} guarantees that, with a large enough number of iterations, system-wide conflict-resolution may be achieved with some probability greater than zero. As per Definition~\ref{Definition - Conflict-free operation}, and focusing on step~3 from Algorithm~\ref{Algorithm - DEMPC coordination algorithm}, once system-wide conflict-free operation has been achieved, the inequality $\breve{V}_i^s \leq \hat{V}_i^s$ is guaranteed to be satisfied after each iteration for any agent~$i \in \mathcal{I}$. Let $\hat{V}_i^{s,p}$ and $\breve{V}_i^{s,p}$ denote $\hat{V}_i^s$ and $\breve{V}_i^s$ computed during iteration number $p$. Since $\hat{V}_i^p$ is computed using optimal state and input trajectories of agent~$i$ that have been updated during iteration $p$, then $\hat{V}_i^{s,p} \leq \breve{V}_i^{s,p - 1}$ must be true during conflict-free operation. Therefore, since $\breve{V}_i^{s,p} \leq \hat{V}_i^{s,p}$ and $\hat{V}_i^{s,p} \leq \breve{V}_i^{s,p - 1}$, then $\breve{V}_i^{s,p} \leq \breve{V}_i^{s,p - 1}$ must be true after conflicts have been resolved, which indicates that the value of $J_i(\cdot)$ computed using updated stationary points of all relevant agents is guaranteed to decrease with each subsequent iteration. If Assumption~\ref{Assumption - Bounded cost function minima} concerning bounded minima holds, then $\left( \mathbf{x}_i^s, \mathbf{u}_i^s \right)$ is guaranteed to approach a local minimizer of $J_i(\cdot)$, thus satisfying Inequality~(\ref{Equation - Stationary state vector convergence}). The above logic may be extended to step~4 of Algorithm~\ref{Algorithm - DEMPC coordination algorithm} without modification; thus concluding the proof.
\end{proof}

\begin{theorem} \label{Theorem - Recursive feasibility}
	(Recursive feasibility) The constraints of Problems~(\ref{Equation - DEMPC - Stationary optimization problem}) and~(\ref{Equation - DEMPC - Dynamic optimization problem}) are guaranteed to be satisfied during each iteration at any sampling time $k \geq 0$ for all agents~$i \in \mathcal{I}$.
\end{theorem}

\begin{proof}
	Provided that initial values of $\overline{\mathbf{x}}_i$, $\overline{\mathbf{u}}_i$, $\mathbf{x}_i^s$, and $\mathbf{u}_i^s$ established at the start of some time-step are feasible, then at least one viable set of values for $\overline{\mathbf{x}}_i$, $\overline{\mathbf{u}}_i$, $\mathbf{x}_i^s$, and $\mathbf{u}_i^s$ exists that satisfies Constraints~(\ref{Equation - DEMPC - Dynamic optimization problem constraints}) and~(\ref{Equation - DEMPC - Stationary optimization problem constraints}), thus ensuring recursive feasibility at each subsequent iteration within the time-step. Further, since the absence of disturbances restricts state progression to $\overline{\mathbf{x}}_i$, then feasibility is preserved at subsequent time-steps. Guaranteeing recursive feasibility thus requires that $\overline{\mathbf{x}}_i$, $\overline{\mathbf{u}}_i$, $\mathbf{x}_i^s$, and $\mathbf{u}_i^s$ are initially feasible. Referring to Eq.~(\ref{Equation - Feasible initial state vector set}), the existence of $\overline{\mathbf{x}}_i$, $\overline{\mathbf{u}}_i$, $\mathbf{x}_i^s$, and $\mathbf{u}_i^s$ such that Constraints~(\ref{Equation - DEMPC - Dynamic optimization problem constraints}) and~(\ref{Equation - DEMPC - Stationary optimization problem constraints}) are satisfied requires that the set $\mathcal{X}_i^{0 \rightarrow s}$ is not empty. Assumption~\ref{Assumption - Controllability} concerning weak controllability bounds the control input trajectory $\overline{\mathbf{u}}_i$ required to steer any initial state vector $\mathbf{x}_i^0 \in \mathcal{X}_i^{0 \rightarrow s}$ to a reachable stationary point $\left( \mathbf{x}_i^s, \mathbf{u}_i^s \right)$ such that $(\mathbf{x}_i^0, \overline{\mathbf{u}}_i, \mathbf{x}_i^s) \in \mathcal{Z}_i^{0 \rightarrow s}$. As a result, the set $\mathcal{X}_i^{0 \rightarrow s}$ must be non-empty; thus concluding the proof.
\end{proof}

\begin{theorem} \label{Theorem - Stability}
	(Closed-loop stability) For all agents~$i \in \mathcal{I}$, as $k \rightarrow \infty$, the measured state vector $\mathbf{x}_i$ will remain bounded within a set $\mathcal{X}_i^* \subset \mathcal{X}_i^{0 \rightarrow s}$ surrounding a fixed stationary point $\mathbf{x}_i^* \in \mathcal{X}_i^s$. The set $\mathcal{X}_i^*$ is defined as follows:
	\begin{equation}
		\mathcal{X}_i^* \coloneqq \left\{ \mathbf{x}_i \in \mathcal{X}_i~|~\exists \overline{\mathbf{u}}_i \in \overline{\mathcal{U}}_i : \left( \mathbf{x}_i, \overline{\mathbf{u}}_i, \mathbf{x}_i^* \right) \in \mathcal{Z}_i^{0 \rightarrow s} \right\}.
	\end{equation}
\end{theorem}

\begin{proof}
	For all agents~$i \in \mathcal{I}$, while updating candidate state and input vector trajectories $\overline{\mathbf{x}}_i$ and $\overline{\mathbf{u}}_i$, Constraint~(\ref{Equation - Terminal state vector constraint}) ensures that the state vector $\mathbf{x}_i$ always remains within a reachable set surrounding some steady state vector $\mathbf{x}_i^s$. Therefore, in order to prove that this reachable set ultimately maintains a fixed value $\mathcal{X}_i^*$, one must prove that the terminal steady state vector $\mathbf{x}_i^s$ approaches a fixed value $\mathbf{x}_i^*$ as $k \rightarrow \infty$ for all agents~$i \in \mathcal{I}$.
	
	We assume for the time being that, after a sufficient number of iterations, agents establish a social hierarchy that permanently resolves all conflicts. That is, Inequality~(\ref{Equation - Conflict-free stationary point negotiation}) is guaranteed to be satisfied at each iteration within all subsequent time-steps, and the social hierarchy therefore ceases to change. We shall refer to such a social hierarchy as a \textit{universal social hierarchy}. Let $\breve{V}_i^{s,k} |_\mathrm{A}$ and $\breve{V}_i^{s,k} |_\mathrm{B}$ denote $\breve{V}_i^s$ computed at the initial and final iterations, respectively, of time-step $k$. It is clear that, after a universal social hierarchy has been established, $\breve{V}_i^{s,k} |_\mathrm{B} \leq \breve{V}_i^{s,k} |_\mathrm{A}$ is satisfied at all subsequent time-steps. Since $\breve{V}_i^{s,k} |_\mathrm{A}$ is computed using $\breve{\mathbf{x}}_i^s$ and $\breve{\mathbf{u}}_i^s$ obtained by solving Problem~(\ref{Equation - DEMPC - Stationary optimization problem}) at the first iteration of time-step $k$, then $\breve{V}_i^{s,k} |_\mathrm{A} \leq \breve{V}_i^{s,k - 1} |_\mathrm{B}$ must also be true. As a result, $\breve{V}_i^{s,k} |_\mathrm{B} \leq \breve{V}_i^{s,k - 1} |_\mathrm{B}$ must be satisfied at all time-steps following the establishment of a universal social hierarchy. If Assumption~\ref{Assumption - Bounded cost function minima} concerning bounded minima holds, $\left( \mathbf{x}_i^s, \mathbf{u}_i^s \right)$ is guaranteed to approach some fixed point $\left( \mathbf{x}_i^*, \mathbf{u}_i^* \right)$ for all $i \in \mathcal{I}$ as $k \rightarrow \infty$, where $\left(\mathbf{x}_i^*, \mathbf{u}_i^*\right)$ is some local minimizer of $J_i(\cdot)$.
	
	The remaining task is to prove that a universal social hierarchy is in fact attainable after a sufficient number of iterations. The existence of at least one such social hierarchy has already been established in the proof for Theorem~\ref{Theorem - Conflict resolution}. Namely, if Assumption~\ref{Assumption - Neighborhood cooperative objectives} concerning neighborhood-cooperative cost functions holds, and if no two agents~$i$ and~$j$ such that $j \in \mathcal{N}_{-i|J}$ for all $i, j \in \mathcal{I}, j \neq i$ make decisions simultaneously, then Inequality~(\ref{Equation - Conflict-free stationary point negotiation}) is guaranteed to be satisfied at every iteration within any sampling time. If conflict persists, then agents will, after a sufficient number of iterations, self-organize according to a universal social hierarchy with some probability greater than zero; thus concluding the proof.
\end{proof}

\subsection{Determining social hierarchy properties} \label{Subsection - Social hierarchy properties}

One final issue that must be addressed concerns the selection of $N_q$. In order to guarantee that a universal social hierarchy is attainable, $N_q$ must be large enough such that a social hierarchy wherein no neighboring pairs of agents operate in parallel is permissible. This goal invokes the vertex coloring problem from graph theory~\cite{Chartrand2008}. In brief, vertex coloring of a graph requires assigning colors to all nodes such that no two interconnected nodes share the same color. Returning to the context of the current paper, each \textit{node} signifies an agent, each \textit{color} represents a specific level along a social hierarchy, and \textit{interconnection} symbolizes cost function coupling. In graph theory, the \textit{chromatic number} refers to the minimum number of colors required to complete the vertex coloring problem. Therefore, in the current context, $N_q$ should be equal to or greater than the chromatic number of the system graph.

\section{Numerical example} \label{Section - Numerical example}

\subsection{Problem description}

Consider the mechanical system described in Fig.~\ref{Figure - Numerical problem setup}. This setup consists of $N$ square plates that are supported by spring-damper systems and perfectly aligned at equilibrium. Each square plate has a mass $m = 1.0\,\si{kg}$, side length $L = 0.25\,\si{m}$, and is supported by stiffness and damping coefficients $k = 1.0\,\si{N/m}$ and $c = 1.0\,\si{kg/s}$. Further, the vertical position of each plate~$i$ is controlled via an input force $u_i$ that is regulated by agent~$i$. The resulting continuous-time dynamics of each plate are expressed in state-space form as follows:
\begin{equation} \label{Equation - Numerical problem model}
	\begin{bmatrix}
		\dot{x}_i \\
		\dot{v}_i
	\end{bmatrix}
	= \begin{bmatrix}
		0 & 1 \\
		-\frac{k}{m} & -\frac{c}{m}
	\end{bmatrix}
	\begin{bmatrix}
		x_i \\
		v_i
	\end{bmatrix}
	+
	\begin{bmatrix}
		0 \\
		\frac{u_i}{m}
	\end{bmatrix},
\end{equation}
where $x_i$ and $v_i$ denote the vertical position and velocity of plate~$i$. For the remainder of the current section, we work with the discrete-time form of Eq.~(\ref{Equation - Numerical problem model}).

\begin{figure}
	\centering
	\includegraphics[width=4in]{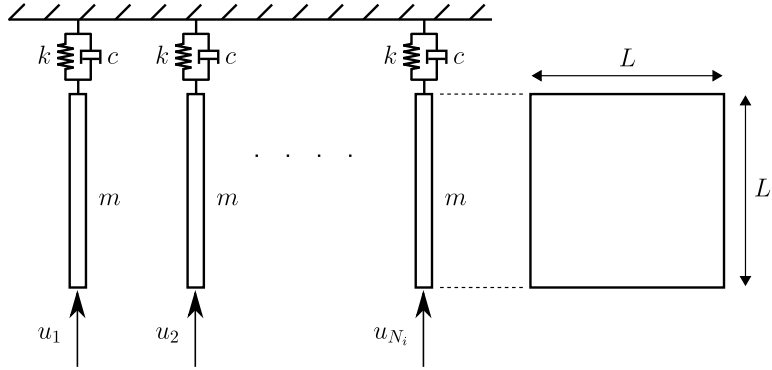}
	\caption{Schematic of the plate overlap problem used in the numerical example.} \label{Figure - Numerical problem setup}
\end{figure}

The economic control objective of each agent~$i$ is to, without excessive actuation, minimize the overlap area between its respective plate and those of neighboring agents~$i - 1$ and $i + 1$. The resulting stage cost function of agent~$i$ is therefore expressed as follows:
\begin{equation} \label{Equation - Numerical problem cost function}
	J_i(x_i, u_i, x_{-i}) = \frac{1}{\left| \mathcal{N}_{-i} \right|} \sum_{j \in \mathcal{N}_{-i}} A_{i,j}(x_i,x_j) + u_i^\mathrm{T} u_i,
\end{equation}
where the set $\mathcal{N}_{-i}$ contains the indices of all neighboring plates~$j$ that are physically adjacent to plate~$i$, $\left| \mathcal{N}_{-i} \right|$ is the cardinality of the set $\mathcal{N}_{-i}$, and $A_{i,j}(\cdot)$ defines the overlap area between agents~$i$ and~$j$ as follows:
\begin{equation}
	A_{i,j}(x_i,x_j) = \left\{
	\begin{matrix}
		0 & , & \left| x_i - x_j \right| \geq L, \\
		L \left( L - \left| x_i - x_j \right| \right)& , & \left| x_i - x_j \right| < L.
	\end{matrix}\right.
\end{equation}
Note that, although simple in terms of system dynamics, the numerical problem described above entails nonlinear and non-convex cost functions. Physically, non-convexity stems from the property that any two adjacent plates may be relocated in multiple ways to minimize their respective overlap areas.

\subsection{Social hierarchy-based DEMPC properties}

The interaction graph for an example problem with ten plates is shown in Fig.~\ref{Figure - Numerical problem graph}. The solution to the vertex coloring problem for this example involves simply alternating the color of each subsequent node, which yields a chromatic number of two regardless of the quantity of vertices. An appropriate choice for the number of hierarchy levels is thus $N_q = 2$.

\begin{figure}
	\centering
	\includegraphics[width=2in]{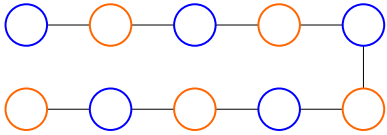}
	\caption{Interaction graph of the plate overlap problem with vertex coloring used to identify a universal social hierarchy.} \label{Figure - Numerical problem graph}
\end{figure}

We implement prediction and control horizons of $H = 5$ time-steps and a sampling period of $1.0\,\si{\sec}$. Within a single sampling period, $N_p = 5$ iterations are permitted for the negotiation of optimal stationary points (\textit{i.e.} step~3 in Algorithm~\ref{Algorithm - DEMPC coordination algorithm}), followed by five more iterations for trajectory optimization (\textit{i.e.} step~4 in Algorithm~\ref{Algorithm - DEMPC coordination algorithm}). Finally, the actuation force of any agent is bounded as $-0.25\,\si{N} \leq u_i \leq 0.25\,\si{N}$. These values were selected to limit the steady-state displacement of each plate to a maximum of $L = 0.25\,\si{m}$.

Finally, all plates are initially at rest with zero displacement from equilibrium (\textit{i.e.} $\mathbf{x}_i = \mathbf{0}$ and $\mathbf{u}_i = \mathbf{0}$) and are therefore perfectly aligned with their neighbors. The state and input vector trajectories of all agents are initialized as $\overline{\mathbf{x}}_i = \left( \mathbf{0}, \cdots, \mathbf{0} \right)$ and $\overline{\mathbf{u}}_i = \left( \mathbf{0}, \cdots, \mathbf{0} \right)$.

\subsection{Simulation results}

\subsubsection{Social hierarchy-based DEMPC -- ten plates}

The first results we present pertain to five simulations involving ten plates. The outcomes in these simulations will differ due to the element of randomness in the proposed coordination algorithm. In Fig.~\ref{Figure - Results - 10 plates, social hierarchy evolution}, we show the variation in social hierarchy levels over the iteration number for all five simulations. Rather than displaying the actual hierarchy level of any particular agent, which would yield a cluttered image, we instead present the cumulative number of social hierarchy level changes. That is, let some cumulative counter start at zero for each simulation. Then, each time an agent alters its social hierarchy level, the cumulative count increases by one.

\begin{figure}
	\centering
	\includegraphics[width=4in]{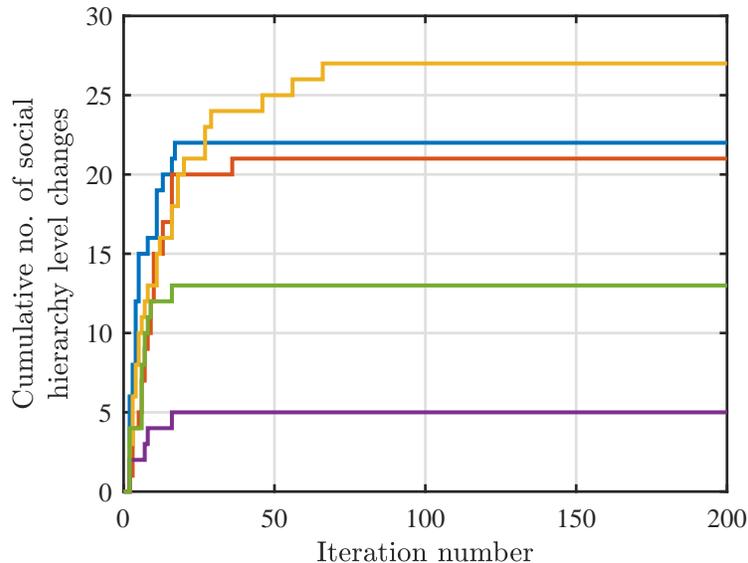}
	\caption{Evolution of the cumulative number of social hierarchy changes for five simulations of ten plates. Each color corresponds to a different simulation.} \label{Figure - Results - 10 plates, social hierarchy evolution}
\end{figure}

What is observed in Fig.~\ref{Figure - Results - 10 plates, social hierarchy evolution} is that, in all five simulations, the cumulative count of social hierarchy variations ultimately reaches a fixed value. This outcome indicates that, with a sufficient number of iterations, the agents sort themselves along a social hierarchy that guarantees conflict-free operation according to Definition~\ref{Definition - Conflict-free operation} in all future iterations, thus validating Theorem~\ref{Theorem - Conflict resolution} concerning conflict resolution.

Next, we plot the evolution of the global cost function $V$ with respect to the iteration number in Fig.~\ref{Figure - Results - 10 plates, global cost function} for all five simulations, with $V$ computed as follows:
\begin{equation}
	V = \sum_{i \in \mathcal{I}} \breve{V}_i.
\end{equation}
Note that the evolution of $V$ differs in each case due to the element of stochasticity inherent to the proposed algorithm. However, in all five simulations, a reduction in $V$ to some locally-optimal value is evident after conflicts have been resolved. Further, within each time-step (\textit{i.e.} within each 20 iteration interval), the reduction or preservation of $V$ is apparent after agents have settled on an appropriate social hierarchy. This latter outcome validates Theorem~\ref{Theorem - Convergence} on convergence.

\begin{figure}
	\centering
	\includegraphics[width=4in]{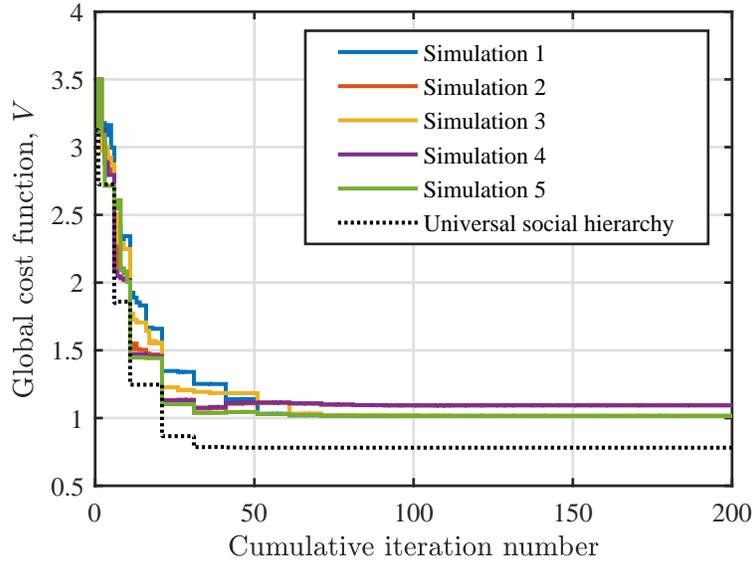}
	\caption{Evolution of the global cost function for five simulations of ten plates.} \label{Figure - Results - 10 plates, global cost function}
\end{figure}

As a reference in Fig.~\ref{Figure - Results - 10 plates, global cost function}, we plot (using a black dotted line) the globally-optimal evolution of $V$, which is obtained by initializing the agents' hierarchy levels according to the universal social hierarchy from Fig.~\ref{Figure - Numerical problem graph}. In this case, the agents do not alter their social hierarchy levels as conflict-free operation is guaranteed from the beginning of the simulation. As a result, $V$ is reduced or preserved within each time-step immediately from the start of the simulation. This result validates the existence of a universal social hierarchy.

In Fig.~\ref{Figure - Results - 10 plates, stationary targets}, we plot the evolution of stationary target positions over the iteration number for five simulations. Rather than displaying the stationary targets of individual agents, which would yield a cluttered image, we plot the mean stationary target $\overline{\mathbf{x}}_s$ of the entire plant, which is computed as follows:
\begin{equation}
	\overline{\mathbf{x}}_s = \frac{1}{N} \sum_{i \in \mathcal{I}} \breve{x}_i^s.
\end{equation}
It is evident from Fig.~\ref{Figure - Results - 10 plates, stationary targets} that values of $\breve{x}_i^s$ for all agents~$i \in \mathcal{I}$ ultimately converge to fixed values in all simulations. This outcome validates Theorem~\ref{Theorem - Stability} concerning bounded closed-loop stability, which requires fixed stationary targets.

\begin{figure}
	\centering
	\includegraphics[width=4in]{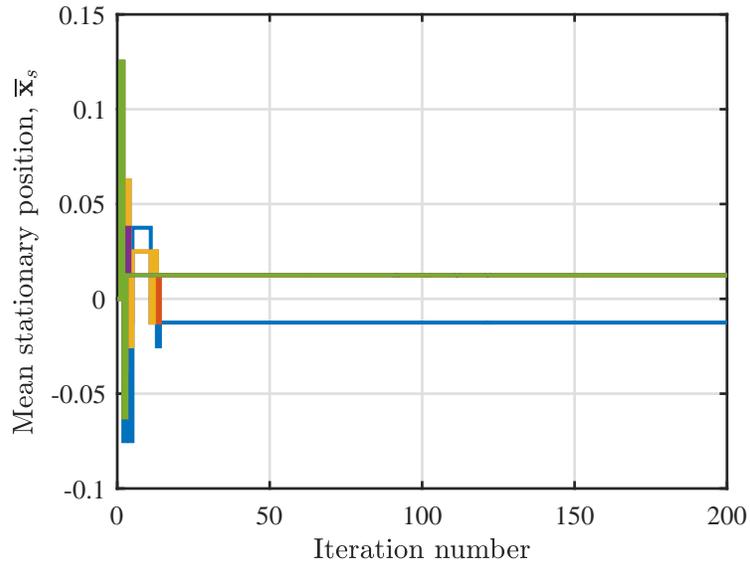}
	\caption{Evolution of the mean system-wide stationary target position for five different simulations involving ten plates.} \label{Figure - Results - 10 plates, stationary targets}
\end{figure}

Finally, we have plotted the locations of all plates at the final sampling time in each simulation in Fig.~\ref{Figure - Results - 10 plates, plate positions}. The color of each plate indicates its social hierarchy level (\textit{i.e.} blue denotes $q_i = 1$, red denotes $q_i = 2$). The globally-optimal configuration would involve all adjacent plates being relocated in opposite directions; however, the social hierarchy-based DEMPC algorithm is only capable of finding a locally-optimal layout wherein some plates (\textit{e.g.} plate~8 in simulation~1) must remain at the origin to minimize overlap with their neighbors.

\begin{figure}
	\centering
	\includegraphics[width=4in]{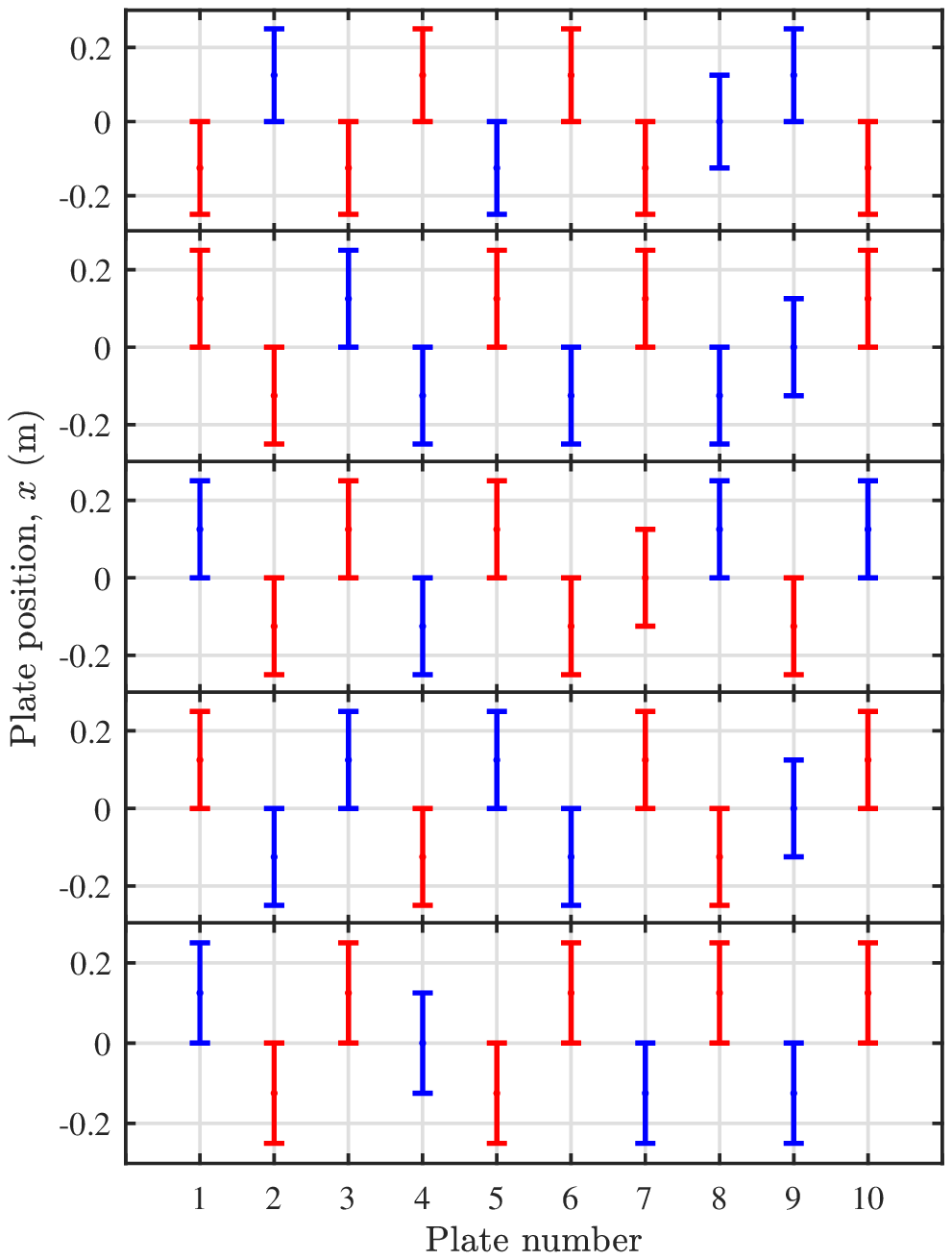}
	\caption{Plate locations at the final sampling time in (from top to bottom) simulations 1 to 5.} \label{Figure - Results - 10 plates, plate positions}
\end{figure} 

\subsubsection{Parallel vs. social hierarchy-based DEMPC}

We compare the performance of the proposed social hierarchy-based DEMPC to a basic parallel DEMPC algorithm wherein all agents solve their local EMPC problems simultaneously and exchange stationary vectors and trajectories with their neighbors. In essence, a parallel DEMPC algorithm is identical to Algorithm~\ref{Algorithm - DEMPC coordination algorithm}, except with $N_q = 1$.

We first plot the evolution of $V$ over the iteration number using a parallel DEMPC algorithm in Fig.~\ref{Figure - Results - Par-DEMPC, Var. plates, global cost function} for five simulations, each consisting of a different number of plates. The parallel DEMPC algorithm is in fact able to naturally resolve conflicts and ultimately decrease $V$ to some locally minimum value. However, as the quantity of agents $N$ increases, a greater number of iterations is required to decrease $V$ to a locally-minimum value. For instance, with $N = 10$, $V$ reaches a local minimum in 20~iterations. If $N = 80$ however, over 100~iterations are required.

\begin{figure}
	\centering
	\includegraphics[width=4in]{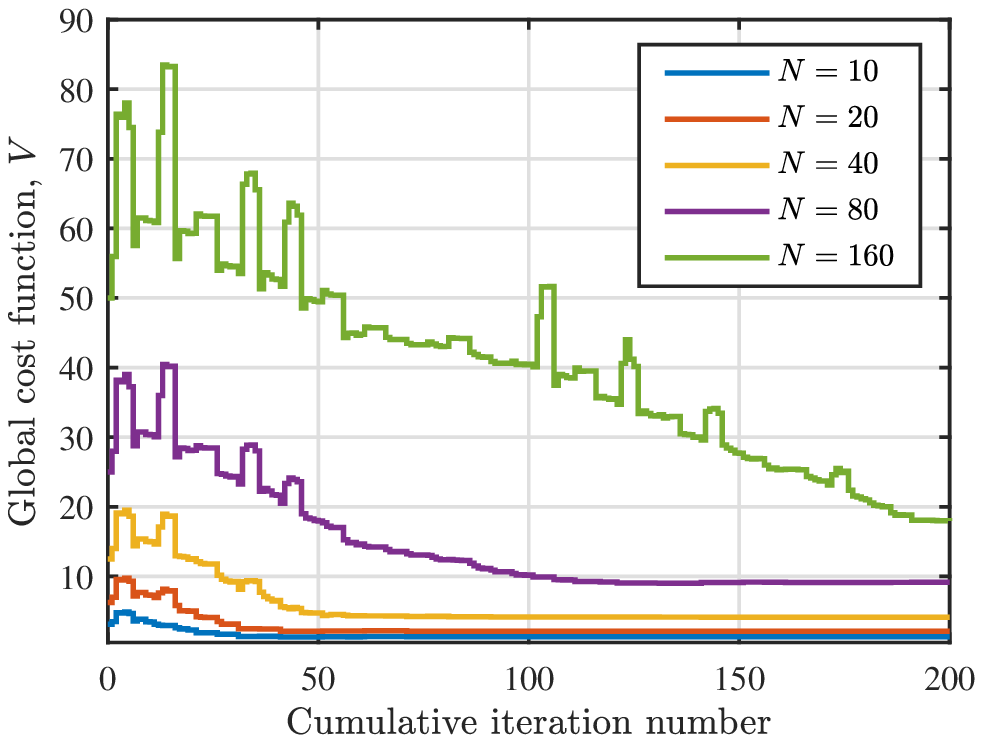}
	\caption{Evolution of the global cost function $V$ using a parallel DEMPC algorithm for five simulations involving different quantities of plates.} \label{Figure - Results - Par-DEMPC, Var. plates, global cost function}
\end{figure}

In Fig.~\ref{Figure - Results - Var. plates, global cost function}, we plot the same information as Fig.~\ref{Figure - Results - Par-DEMPC, Var. plates, global cost function}, except using the proposed social hierarchy-based DEMPC algorithm. The improvement is clear. The value of $N$ has no discernible effect on the number of iterations necessary for reducing $V$ to a locally-minimum value. In each simulation, approximately 30~iterations are required for $V$ to settle at some minimum value. This outcome results from the fact that the likelihood of any agent resolving conflict locally is dependent solely on its neighborhood interaction topology. Thus, raising $N$ should not impact the number of iterations required for system-wide conflict resolution. This outcome is further validation of the scalability of the social hierarchy-based method.

\begin{figure}
	\centering
	\includegraphics[width=4in]{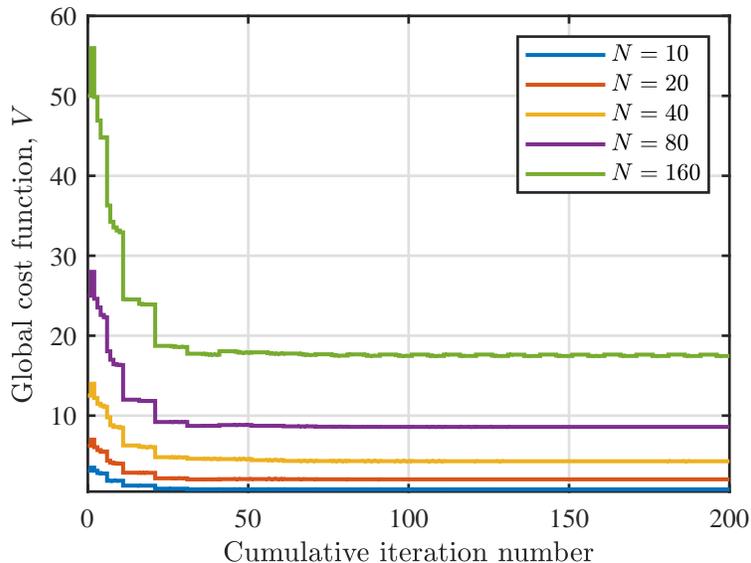}
	\caption{Evolution of the global cost function $V$ using the proposed social hierarchy-based DEMPC algorithm for five simulations involving different quantities of plates.} \label{Figure - Results - Var. plates, global cost function}
\end{figure}

\section{Conclusions and recommendations for future research} \label{Section - DEMPC - Conclusions}

We have presented a novel concept for addressing non-convexity in cost functions of distributed economic model predictive control systems with unknown terminal stationary targets. This concept involves agents self-organizing into a finite hierarchy using evolutionary principles, and ultimately enables agents to make decisions that are mutually beneficial with those of their neighbors. Theorems guaranteeing convergence, recursive feasibility, and bounded closed-loop stability have also been provided for the proposed social hierarchy-based algorithm.

These theorems were validated using a numerical example involving a series of suspended square plates wherein each agent attempted to minimize the overlap area between its respective plate and those of its neighbors. Results showed that, across five simulations, the proposed algorithm was capable of establishing a social hierarchy that reduced the system-wide cost function to some locally-minimum value. Another observation from numerical results was that increasing the size of the distributed system (\textit{i.e.} the number of plates and agents) had no discernible effect on the number of iterations required to minimize cost function values to local minima. This behavior was not observed when using a parallel DEMPC algorithm with no mechanism for addressing non-convexity.

Several research directions exist for further enhancing the proposed DEMPC algorithm; developing non-iterative algorithms using compatibility constraints as first proposed by Dunbar and Murray~\cite{Dunbar2006}; employing Lyapunov constraints to guarantee asymptotic stability rather than bounded stability; guaranteeing convergence, feasibility, and stability under the effects of coupled dynamics and constraints; ensuring robustness in the presence of bounded disturbances in the system dynamics and cost functions; application of the proposed algorithm to distributed systems with weak dynamic coupling such as autonomous vehicle trajectory planning and wind farm control.

\section*{Acknowledgment}

The authors are grateful for the financial support provided by the Natural Sciences and Engineering Research Council of Canada (NSERC).

\bibliographystyle{unsrt}
\bibliography{Library.bib}

\end{document}